\documentclass{preprint}
  \usepackage{mhfig}
  \usepackage{amsmath}
  \usepackage{amssymb}
  \usepackage{times}%
  \usepackage{mhequ}
  \usepackage{mhenvs}

  \usepackage{ScriptFonts}
  \usepackage{Symbols}
  \newcommand{\one}{\mathbf{1}}
  \newcommand{\E}{\mathbf{E}}
  \newcommand{\hf}{{\textstyle{1\over 2}}}
  \renewcommand{\P}{\mathbf{P}}
  \newcommand{\OU}{{\mbox{\tiny OU}}}
  
  \def\les{\lesssim}
  \def\hf{{\textstyle{1 \over 2}}}
  \def\Hf{H_{\! f}}

\begin{document} 

  \date{December 21, 2007}
  \title{Slow energy dissipation in \\[1mm] anharmonic oscillator chains}
  \author{Martin~Hairer\inst1 and Jonathan C.~Mattingly\inst2}
  \institute{Mathematics Institute, The University of Warwick, Coventry CV4 7AL, UK
   \\ \email{M.Hairer@Warwick.ac.uk}
  \and Department of Mathematics and Center for Nonlinear and Complex Systems, Duke University, Durham 27701, USA
   \\ \email{jonm@math.duke.edu}}
  \titleindent=0.65cm

  \maketitle
  \thispagestyle{empty}

  \begin{abstract}
    \parindent1em
    We study the dynamic behavior at high energies of a chain of
    anharmonic oscillators coupled at its ends to heat baths at possibly
    different temperatures. In our setup, each oscillator is subject to a homogeneous
    anharmonic pinning potential $V_1(q_i) =|q_i|^{2k}/2k$ and harmonic 
    coupling potentials $V_2(q_i- q_{i-1}) = (q_i- q_{i-1})^2/2$ between itself and its
    nearest neighbors. 
    We consider the case $k > 1$ when the pinning
    potential  is
    stronger then the coupling potential. At high energy, when a large fraction of the energy
    is located in the bulk of the chain, breathers appear and block the transport of energy
    through the system, thus slowing its convergence to equilibrium.

    In such a regime, we obtain equations for an effective dynamics by averaging out the
    fast oscillation of the breather. Using this representation and related
    ideas, we can prove a number of results. When the chain is of length
    three and $k> 3/2$ we show that there exists a unique invariant
    measure. If $k > 2$ we further show that the system does not
    relax exponentially fast to this equilibrium by demonstrating that
    zero is in the essential spectrum of the generator of the dynamics.  When the chain
    has five or more oscillators and $k> 3/2$ we show that the generator
    again has zero in its essential spectrum.

    In addition to these rigorous results, a theory is given for the rate
   of  decrease of the energy when it is concentrated in one of the
    oscillators without dissipation. Numerical simulations are included
    which confirm the theory.
  \end{abstract}

  \section{Introduction}

  One subject that has received considerable attention in recent years
  is the return to equilibrium of systems arising from statistical
  mechanics. One of the simplest models of interest is given
  by the kinetic Fokker-Planck equation
  \begin{equ}[e:FP]
    \d_t \phi_t = \CL \phi_t\;,\quad \phi_0 = \phi\;,\quad \CL = {1\over
      2} \d_p^2 - p\, \d_p + p\, \d_q - \nabla V(q)\,\d_p\;,\quad p,q
    \in \R^n\;.
  \end{equ}
  This equation describes the evolution of an observable $\phi \colon
  \R^n \to \R$ under the Hamiltonian dynamic for the energy $H(p,q) =
  {p^2 \over 2} + V(q)$, perturbed by friction and noise:
  \begin{equ}[e:dynFP]
    dp = -\nabla V(q)\, dt -p\, dt + dw(t)\;,\qquad dq = p\,dt\;.
  \end{equ}
  The relation between \eref{e:dynFP} and \eref{e:FP} is given by the
  fact that the function $\phi_t(p_0, q_0) = \E \phi_0(p(t), q(t))$
  satisfies the partial differential equation \eref{e:FP}, provided that
  the pair $(p(t),q(t))$ solves \eref{e:dynFP}.  It can easily be
  checked by inspection that if $V$ is sufficiently smooth and coercive,
   the measure $\mu = \exp(-2H(p,q))\,dp\,dq$
  is invariant under this dynamic in the sense that if $\phi_t$
  satisfies \eref{e:FP}, then $\int \phi_t d\mu$ does not depend on $t$.
  This can be rephrased as saying that if $(p_0,q_0)$ is a random
  variable with law $\mu$ independent of the driving noise $w$, then the law of $(p(t),
  q(t))$ is given by $\mu$ for all times.

  Underpinning much of the analysis of $\CL$ on the space $\L^2(\mu)$ is
  the guiding principle that it related to the corresponding \textit{Witten
  Laplacian} \cite{HelNie05HES}
  \begin{equ}[e:WL]
    \Delta_V = - \Delta_q + |\nabla V(q)|^2 - \bigl(\Delta V\bigr)(q)\;.
  \end{equ}
  In particular, it was conjectured by Helffer and Nier
  \cite[Conjecture~1.2]{HelNie05HES} that $\CL$ has compact resolvent on
  $L^2(\mu)$ if and only if $\Delta_V$ has compact resolvent on the \textit{flat}
  space $\L^2(\R^n)$. This conjecture has been partially solved in
  \cite{HelNie05HES} (see also \cite{Nie06}) in the sense that one can
  exhibit a large class of potentials for which it holds (loosely,
  $V$ should grow in a sufficiently regular way at infinity).
  Recently, upper as well as lower bounds on the spectral gap of $\CL$
  have been obtained in \cite{HerNie02IHA,Vil04HPI} for potentials $V$
  that satisfy certain growth conditions at infinity.

  All of the results cited above make heavy use of the following two key
  facts:
  \begin{claim}
  \item[a.] There is an explicit formula for the invariant measure of
    \eref{e:dynFP}.
  \item[b.] The friction term $-p\d_p$ acts on all (physical) degrees of
    freedom of the system.
  \end{claim}

  Both of these facts are violated in the following very popular model
  for heat conduction.  Take a finite collection of $N+1$ anharmonic
  oscillators with nearest-neighbor couplings, that is a classical
  Hamiltonian system with Hamiltonian given by
  \begin{equ}[e:defH]
    H(p,q) = \sum_{i=0}^{N} \Bigl({p_i^2 \over 2} + V_1(q_i) \Bigr) +
    \sum_{i=1}^{N} V_2(q_i - q_{i-1})\;.
  \end{equ}
  Here, the potential $V_2$ is the interaction potential between
  neighboring oscillators, whereas $V_1$ is a \textit{pinning potential}. This
  system is then put in contact with two heat baths at
  \textit{different} temperatures $T_0$ and $T_{N}$. We model the
  interaction with the heat baths by the standard Langevin dynamics, so
  that the equations of motion of our system are given by
  \begin{equs}
    dp_0 &= - \gamma_0 p_0\,dt - V_1'(q_0)\,dt - V_2'(q_0 - q_1)\,dt +
    \sigma_0\,dw_0\\
    dp_i &= -V_1'(q_i)\,dt - V_2'(q_i - q_{i-1})\,dt - V_2'(q_i -
    q_{i+1})\,dt  \label{e:main}\\
    dp_{N} &= - \gamma_{N} p_{N}\,dt - V_1'(q_{N})\,dt -
    V_2'(q_{N} - q_{N-1})\,dt
    + \sigma_{N}\,dw_{N}\\
    dq_j &= p_j\,dt\;.
  \end{equs}
  Here, we set $\sigma_i^2 = 2\gamma_i T_i$,
  the index $i$ runs from $1$ to $N-1$ and the index $j$ runs from
  $0$ to $N$.
  As described in \cite{LebLuc}, the rigorous analysis of this model and in particular the
  derivation of Fourier's law (or the proof of its breakdown) is an outstanding mathematical
  problem of great interest to the applied community.

  If $T_0 = T_{N} = T$, then one can check as before that this set of
  equations has a unique invariant measure, which has density
  $\exp(-H(p,q)/T)$ with respect to the Lebesgue measure, where $H$ is
  the Hamiltonian from \eref{e:defH}.  When the two temperatures $T_0$
  and $T_{N}$ are different however, much less is known.  In particular,
  as we will see immediately, even the \textit{existence} of an
  invariant probability measure is an open problem in some cases as
  simple as $V_1(q) = q^4$ and $V_2(q) = q^2$.

  This model has been the subject of many studies, both from a numerical
  and from a theoretical perspective. The purely harmonic case has been solved 
  explicitly in \cite{RLL}. Though no explicit solution is known in the anharmonic case,
  a wealth of numerical experiments exist, see for example \cite{LL} and references 
  therein. 
  Since we will mainly focus on the
  theoretical aspects of the model, we refer to
  \cite{EPR99EPN,EPR99NES}, which seems to be one of the first rigorous
  studies of the anharmonic case.  It was shown in
  \cite{EckHai00NES,EckHai03SPH} that if $V_1(q)$ and $V_2(q)$ behave
  approximately like $|q|^{a_1}$ and $|q|^{a_2}$ respectively at
  infinity then, provided that $a_2 > a_1 > 2$, there exists a unique
  invariant measure for \eref{e:main}. The statement that was proved was
  actually stronger than that, namely it was shown that the generator
  $\CL$ of \eref{e:main} has compact resolvent in every space of the
  form $L^2(\exp(-H(p,q)/T)\,dp\,dq)$ with $T > \max\{T_0,T_{N}\}/2$.  In
  \cite{ReyTho02ECT}, it was also shown, using entirely probabilistic
  rather than functional-analytic techniques, that the condition $a_2
  \ge a_1 \ge 2$ is sufficient for the existence and uniqueness of an
  invariant measure for \eref{e:main}. Furthermore, the compactness of
  the corresponding semigroup in some weighted $L^\infty$ space was also
  proved there.

  This left open the case $a_2 < a_1$ which is the subject of the
  present study. To our knowledge, no previous rigorous results exist in this case,
  although some interesting theory has been developed recently in \cite{BricKup,lefevere-2005}.
  At first sight, one might think that there is no
  \textit{a priori} reason for the behavior of \eref{e:main} to differ
  in any essential way from the case $a_2 \ge a_1$ where spectral gap
  results are known.  Such wishful thinking turns out to be overly
  optimistic. Even in the simplest possible scenario, that is when
  $V_1(q) = q^4$ and $V_2(q) = q^2$, we will show in
  Theorem~\ref{theo:threeosc} below that the compactness property of the
  resolvent of $\CL$ is destroyed as soon as $N+1 \geq 3$. Furthermore, when
  $N+1 \geq 5$, it will be shown in Theorem~\ref{theo:compactchain} that the
  essential spectrum of $\CL$ (always in a weighted $L^2$ space of the
  type considered before) extends all the way to $0$. These negative
  results hold even in the case where $T_0 = T_{N} = T$, showing that
  having the friction acting on all physical degrees of freedom is a
  crucial assumption for the Helffer-Nier conjecture to hold.

  The reason why the behavior of \eref{e:main} changes so drastically
  when $a_2 < a_1$ can be understood heuristically by the appearance of
  \textit{breathers} (see for example \cite{Breathers}). Breathers are
  dynamically stable, spatially localized, periodic orbits that arise in
  the noise-free  translation invariant ($i = -\infty,\ldots, +\infty$) version of \eref{e:main}.
  Good approximations to these orbits persist in \eref{e:main},
  especially if $N$ is large.  It is therefore possible to put the
  system into a state where most of its energy is localized in a few
  oscillators located in the middle of the chain. On the other hand,
  energy can be dissipated only through the terms $-\gamma p$ appearing
  in the equations for the first and the last oscillator.  Therefore,
  one expects the energy of the system to decay extremely slowly. The
  appearance of breathers can be proved in the case where the strength
  of the nearest-neighbor coupling is much weaker than the strength of
  the pinning potential. At high energies, this is precisely the case
  when $a_2 < a_1$.

  This discussion shows that even the mere \textit{existence} of an
  invariant measure is a non-trivial problem in this model unless $T_0 =
  T_{N}$, in which case one can check explicitly that the usual
  Boltzmann-Gibbs distribution is invariant.  This differs from many other systems
   in that the \textit{existence} of  an invariant
  measure, not its \textit{uniqueness}, is difficult.  
  In fact, the
  uniqueness of an invariant measure for a chain of arbitrary length
  follows quickly from the hypoellipticity of the generator and the Hamiltonian
  structure once the
  existence of an invariant measure is established.

  We are at the moment
  unable to provide a general proof that shows the existence of an
  invariant measure for a chain of arbitrary length. However, in the
  case of a chain comprising of $3$ oscillators, we show in
  Theorem~\ref{theo:wholething} below that there exists a unique
  invariant measure, provided that the coupling potential is harmonic
  and the pinning potential is homogeneous of sufficiently high degree.
  This result is proven by first obtaining an effective dynamics when
  most of the energy is concentrated in the central oscillator. This
  effective dynamics is then used to construct a Lyapunov function whose
  level sets are compact and hence, by a variation on the classical
  Kryloff-Bogoliouboff method, implies the existence of an invariant measure.

  The remainder of this article is organized as follows. In
  Section~\ref{sec:formalCalc}, we give a formal calculation that shows
  how it is possible to relate the spectral properties of $\CL$ to the
  scaling properties of the potentials $V_1$ and $V_2$.  The results
  given by these formal calculations are then compared to numerical
  simulations. We proceed in Section~\ref{sec:negative} to the
  proof of the \textit{negative} results concerning the lack of compactness and/or 
  of a spectral gap for $\CL$. In  Section~\ref{sec:effective} we derive effective equations of motion
  for the system of three oscillators in the regime when a breather is present. Here, we
  make heavy use of the compensator techniques from averaging / homogenization
  theory \cite{BLP}. Finally, we show in
  Section~\ref{sec:existence} that in the simplest case of a chain of
  three oscillators with harmonic coupling potentials, one can show the
  existence of an invariant measure for any degree of homogeneity
  greater than $2$ for the pinning potential. Perhaps more interesting
  then this last result is the method of proof. We derive as system of
  effective equations and prove their accuracy that when the energy of
  middle oscillator is large. These effective equations give insight
  into the mechanism of energy dissipation and are used to construct a
  Lyapunov function.

  \begin{acknowledge}
    We would like to thank the University of Bonn (and the then-nascent
    Hausdorff institute) for providing the pleasant working environment
    where this project began. MH would like to thank Jean-Pierre Eckmann and
    Luc Rey-Bellet for many long discussions about this problem. We would also like to thank Arnaud Guillin
    for his encouragements which aided in completing this text in a
    (more or less) timely manner.  MH acknowledges support from EPSRC
    fellowship EP/D071593/1. JCM acknowledges the support of the Sloan
    Foundation and National Science Foundation through a PECASE award
    (DMS 04-49910) and a standard award (DMS 06-16710). 
  \end{acknowledge}

  \section{A formal calculation}
  \label{sec:formalCalc}
  In this section, we first perform a formal calculation that allows us
  to get a feeling of the speed at which energy is extracted from such a
  system. Consider the simplest non-trivial case, that is when $N+1 = 3$.
  In order to keep things simple, we furthermore assume that the
  coupling potential $V_2$ is quadratic: $V_2(q) = q^2/2$ and that the
  pinning potential is homogeneous of degree $2k$: $V_1(q) =
  |q|^{2k}/(2k)$ for some real number $k > 1$. The equations of motion
  for the system of interest are thus given by
  \begin{equs}[e:threeosc]
    dp_0 &= - \gamma_0 p_0\,dt - q_0 |q_0|^{2k-2}\,dt - (q_0 - q_1)\,dt
    + \sqrt{2\gamma_0 T_0}\,dw_0\\
    dp_1 &=  - q_1 |q_1|^{2k-2}\,dt - (2q_1 - q_0 - q_2)\,dt\\
    dp_2 &= - \gamma_2 p_2\,dt - q_2 |q_2|^{2k-2}\,dt - (q_2 - q_1)\,dt
    + \sqrt{2\gamma_2 T_2}\,dw_2\\
    dq_j &= p_j\,dt\;.
  \end{equs}
  Let us first have a look at the motion of the middle oscillator by
  itself, \ie at the solution of
  \begin{equ}[e:simple]
    dp = -q |q|^{2k-2}\,dt \;,\qquad dq= p\,dt\;.
  \end{equ}
  It is easy to see that this equation is invariant under the
  substitution
  \begin{equ}[e:rescaled]
    q(t) = E^{1\over 2k}\, \tilde q(E^{\alpha}t)\;,\quad p(t) = E^{1 \over
      2}\,
    \tilde p(E^{\alpha}t)\;,
  \end{equ}
  with $\alpha = {1\over 2} - {1\over 2k}$.  Let us now denote by
  $(\tilde p, \tilde q)$ the solution to \eref{e:simple} with initial
  condition $(\sqrt{2},0)$, so that $(p,q)$ as given by
  \eref{e:rescaled} is the (unique up to a phase) solution to
  \eref{e:simple} at energy $E$. Since the variables $p$ and $q$ are
  assumed to be one-dimensional, the solution $(\tilde p, \tilde q)$ is
  periodic, say with period $\tau$.

  Consider now the equation for the left oscillator, into which we
  substitute the (approximate) solution to the motion of the middle
  oscillator that we just found:
  \begin{equs}
    dp_0 &= - \gamma_0 p_0\,dt - q_0 |q_0|^{2k-2}\,dt + q_0\,dt +
    \sqrt{2 \gamma_0 T_0} dw_0 + E^{1\over 2k}\,\tilde q(E^{\alpha}t)\,dt\\
    dq_0 &= p_0\,dt \label{e:firstosc}\;.
  \end{equs}
  If $E$ is large compared to the typical size of $(p_0, q_0)$ we expect
  that, up to lower order corrections, the solution to this equation
  behaves like the superposition of the solution $(\bar p_0, \bar q_0)$
  to \eref{e:firstosc} with the exogenous forcing $\tilde q \equiv 0$
  and of a highly oscillatory term of the form
  \begin{equ}[e:osccontr]
    \tilde p_0 = E^{\frac1{2k} - \alpha}\, P(E^{\alpha}t)\;,\qquad \dot P =
    \tilde q\;.
  \end{equ}
  By symmetry, the same applies to the right oscillator $(p_2, q_2)$.
  Applying now It\^o's formula to the total Hamiltonian $H$ for
  \eref{e:threeosc} yields
  \begin{equ}[e:behaveHFirst]
    {d \over dt} \E H(t) = \E \bigl(\gamma_0 T_0 + \gamma_2 T_2 -
    \gamma_0 p_0^2(t) - \gamma_2 p_2^2(t)\bigr)\;.
  \end{equ}
  In light of the above discussion, we take $(p_i,q_i)\approx(\bar p_i +
  \tilde p_i,\bar q_i + \tilde q_i)$ for $i=0,2$ so $p_i^2 = \bar p_i^2
  + 2 \bar p_i \tilde p_i + \tilde p_i^2$. From its definition, observe
  that $\E \bar p_i^2(t) \to T_i$ as $t \to \infty$ for $i=0,2$.
  Furthermore, when the energy $E$ of the middle oscillator is large, we
  expect the product $\bar p_0(t) \tilde p_0(t)$ to average out to a small quantity
  when integrated over time periods much larger than $E^{-\alpha}$ due
  to the highly oscillating, mean zero $\tilde p_0(t)$. Applying this
  line of reasoning to  \eqref{e:behaveHFirst} shows that, in the regime where most of the
  energy of the system is concentrated in the middle oscillator, one
  expects to have
  \begin{equ}[e:behaveH] 
    {d \over dt}  H(t) \approx - (\gamma_0 \E \tilde p_0^2+ \gamma_2\E
    \tilde p_2^2) \approx - (\gamma_0 + \gamma_2) \kappa_k H(t)^{{2\over k} -
      1}\;,
  \end{equ}
  where $\kappa_k$ is the variance of the function $P$ introduced in
  \eref{e:osccontr}, that is
  \begin{equ}[e:defkappak]
    \kappa_k = {1\over \tau}\int_0^\tau P^2(s)\,ds - {1\over \tau^2}
    \Bigl(\int_0^\tau P(s)\,ds\Bigr)^2\;,\qquad \dot P(t) = \tilde
    q(t)\;.
  \end{equ}
  The dependence on $k$ comes from the fact that $\tilde q$ is the position
  of the free oscillator with potential ${|q|^{2k}\over 2k}$.
  Actually, one expects this behavior to be correct even in a chain
  with more than just three oscillators.  If the chain has $N+1 = 2n+1$
  oscillators and most of the energy is stored in the middle oscillator,
  one expects the motion of the endpoints to be given by the
  superposition of a \textit{slow} motion and a highly oscillatory
  \textit{fast} motion $\tilde p_0$ with a scaling of the type
  \begin{equ}
    \tilde p_0 = E^{\frac1{2k} - (2n-1)\alpha} \,\hat P(E^\alpha t)\;,
  \end{equ}
  where $\hat P$ is the (unique) periodic function such that
  ${d^{2n-1}\hat P \over dt^{2n-1}} = \tilde q$ and such that the
  integral of $\hat P$ over one period vanishes. (This is because we
  assume that the nearest-neighbor couplings are linear.) This suggests
  that in the general case of a chain of length $N+1 = 2n+1$, the energy
  of the system decreases like
  \begin{equ}[e:behaveHlong] {d \over dt} H(t) \approx - (\gamma_0 +
    \gamma_{N}) \kappa_{k,n} H(t)^{{2n\over k} + 1 - 2n}\;,
  \end{equ}
  for $\kappa_{k,n}$ the variance of $\hat P$ (of course, $\kappa_{k,1} = \kappa_k$ as defined above). 
  If $N+1 = 2n$ is even, the worst-case
  scenario is obtained by storing most of the energy in one of the two
  middle oscillators. Since their distance to the boundary is the same
  as the distance of the middle oscillator to the boundary in the chain
  of length $2n+1$, we expect the rate of decay of the energy to be
  similar in both cases.

  \begin{remark}
    If the coupling potential is not quadratic but homogeneous of degree
    $2\ell$, one can still perform a calculation similar to the one we
    just did, but one has to be more careful. When looking at the
    influence of $q_i$ on $p_{i-1}$ say, one should take into account
    whether the fast oscillations of $q_i$ are of order smaller or
    larger than $1$. If they are of order smaller than one, one can
    linearize the coupling potential. If they are or order larger than
    one, one should multiply them by the scaling exponent arising in the
    coupling. Suppose as before that the chain contains $N+1 = 2n+1$
    oscillators and that most of it's energy is stored in the middle
    oscillator (oscillator $n$).  Assume that the amplitudes in the fast
    oscillations of $p_i$ and $q_i$ scale like $E^{\beta_i}$ and
    $E^{\beta_i-\alpha}$ respectively. (Recall that $\alpha = {1\over 2}
    - {1\over 2k}$ is the exponent giving the period of the
    oscillations.)

  One then has $\beta_n = 1/2$. The values of $\beta_i$ with $i < n$ are given by the
  following recursion formula:
  \begin{equ}
    \beta_i = \left\{\begin{array}{rl} (2\ell-1)(\beta_{i+1} - \alpha) -
        \alpha & \text{if $\beta_{i+1} > \alpha$} \\
        \beta_{i+1} - 2\alpha & \text{if $\beta_{i+1} \le \alpha$}
      \end{array}\right.
  \end{equ} 
  Using this formula, one can then compute $\gamma = 2\beta_0$.  Note
  that if $\ell = 1$, one obtains $\beta_0 = {1\over 2} - 2n\alpha$
  which agrees with the value for $\gamma$ obtained previously.
  \end{remark}

  \begin{remark}
    One would expect these scaling relations to hold at high energies,
    even if the potentials are not exactly homogeneous. One can then
    still perform most of the analysis presented below by splitting the
    right-hand side of the equations into a homogeneous part and a
    remainder term and by assuming that the remainder gets small (in a
    suitable sense) at high energies.
  \end{remark}

  \subsection{Numerical simulations}

  In this section, we show that there is a surprisingly good
  agreement, even over extremely long time intervals, between numerical
  simulations of \eref{e:threeosc} and the predictions \eref{e:behaveH}
  and \eref{e:behaveHlong}.  In order to compare the two, we introduce
  the function
  \begin{equ}
    H^{(k,n)}(p,q) = H^\beta(p,q)\;,\qquad \beta = 2n \Bigl(1 - {1\over
      k}\Bigr)\;.
  \end{equ}
  Plugging this into \eref{e:behaveHlong}, we get the prediction
  \begin{equ}[e:prediction]
    H^{(k,n)}(t) \approx H^{(k,n)}(0) - (\gamma_0 +
    \gamma_2)n\bigl(2-{\textstyle{2\over k}}\bigr) \kappa_{k,n}\;.
  \end{equ}
  A straightforward numerical simulation (essentially integration of the free equation) 
  furthermore allows to compute
  the values of $\kappa_{k,n}$ to very high precision. For example, we obtain
  \begin{equ}[e:valKn]
  \kappa_{2,1} \approx 0.63546991\;,\qquad
  \kappa_{3,1} \approx 0.42363371\;.
  \end{equ}
  We performed numerical simulations for the cases $k=2$ and $k=3$. Both
  simulations were performed using a modification of the classical
  St\"ormer-Verlet method (see for example \cite{GNI}) to take into
  account for the friction and the noise. The modification was done in
  such a way that the resulting method is still of order two.

  The simulation for $k = 2$ was performed at a stepsize $h = 10^{-3}$,
  and the simulation for $k = 3$ was performed at a stepsize $h = 4\cdot
  10^{-4}$. Both simulations used $\gamma_0 = \gamma_2 = 1.3$ and $T_0 =
  T_2 = 1$.

  \begin{center}
    \mhpastefig(figures/){figure4}
    \qquad
    \mhpastefig(figures/){figure6}
  \end{center}

  The prediction obtained from \eref{e:prediction} with the values \eref{e:valKn}
  is shown as a dashed
  line on these figures, but it fits the numerics so well that it is
  nearly invisible. We emphasize that  there were \textit{no} free parameters in the fit, all constants
  are predicted by the theory.
  Note that the timescale in the second picture differs by a factor $40$ from the
  timescale in the first picture.

  \subsection{Comparison to a gradient diffusion}

  In this section, we argue that the energy decay rate predicted by
  \eref{e:behaveHlong} also yields a prediction on the qualitative
  nature of the spectrum of the generator of \eref{e:threeosc} in the
  weighted space $L^2(\mu)$, where $\mu$ is the invariant measure.

  The idea is to model the behavior of the energy $H(t)$ by a
  one-dimensional diffusion of the type
  \begin{equ}[e:oneD]
    dx = \bigl(b(x)+ a'(x)a(x)\bigr)\, dt + a(x)\, dw(t)\;.  
  \end{equ}
  It is well-known that the invariant measure for \eref{e:oneD} is given by
  \begin{equ}
    \mu(dx) = Z^{-1} \exp \Bigl(2\int_0^x {b(t) \over
      a^2(t)}\,dt\Bigr)\;.
  \end{equ}
  Since we expect the invariant measure of \eref{e:threeosc} to behave
  roughly like $\exp(-\beta H)\,dH$ (up to lower-order corrections), we
  should choose $a$ such that $a^2(x) \approx |b(x)|$ for large $x$.
  Combining this with \eref{e:behaveHlong}, we obtain the model
  \begin{equ}[e:model]
    b(x) = - x^{\gamma} \;,\qquad a(x) = x^{\gamma/2}\;,\qquad \gamma =
    {2n\over k} + 1-2n\;,
  \end{equ}
  which has $2\exp(-2x)\,dx$ as its invariant measure (we restrict
  ourselves to the half-space $x \ge 0$). Note that since $k > 1$ (the
  pinning potential grows faster than the coupling potential), one has
  always $\gamma < 1$.

  With the choice \eref{e:model}, the generator for \eref{e:oneD} is then given by
  \begin{equ}
    \bigl(\CL f\bigr)(x) = {1\over 2} \d_x \bigl(x^\gamma \d_x
    f\bigr)(x) - x^{\gamma} \d_x f(x)\;.
  \end{equ}
  Since the operator $(Kf)(x) = f(x) \exp(-x)$ is a unitary operator
  from the weighted space $L^2(\exp(-2 x)\,dx)$ to the flat $L^2$ space,
  the operator $\CL$ is unitarily equivalent to the operator $\CL_1$ on
  $L^2$ given by
  \begin{equ}
    \bigl(\CL_1 f\bigr)(x) = {1\over 2} \d_x \bigl(x^\gamma \d_x
    f\bigr)(x) - {x^{\gamma -1}\over 2}(x - \gamma)\;.
  \end{equ}
  At this point, we recall that if $\phi\colon \R_+ \to \R_+$ is a
  strictly increasing differentiable function with $\phi(0) = 0$ and
  $\lim_{x \to \infty} \phi(x) = \infty$, then the operator
  \begin{equ}
  \bigl(U_\phi f\bigr)(x) = {f(\phi(x)) \sqrt{\phi'(x)}}
  \end{equ}
  is a unitary operator from $L^2(\R_+)$ to itself which furthermore
  satisfies the identity $U_\phi^{-1} = U_{\phi^{-1}}$. Under
  conjugation with $U_\phi$, we see that one has the unitary
  equivalences
  \begin{equ}
    \d_x \approx {1\over \phi'(x)}\d_x - {\phi''(x) \over 2
      (\phi'(x))^2} \;,\qquad V(x) \approx V(\phi(x))\;.
  \end{equ}
  Choosing $\phi(x) = x^{2/(2-\gamma)}$, we see that $\CL_1$ is
  unitarily equivalent to the operator $\CL_2$ given by
  \begin{equs}
    \bigl(\CL_2 f\bigr)(x) &= {1\over2} \Bigl({2-\gamma\over 2}\d_x +
    {\gamma \over 4x}\Bigr) \Bigl({2-\gamma\over 2}\d_x - {\gamma \over
      4x}\Bigr) - {x^{2\gamma-2\over 2-\gamma}\over 2} \bigl(x^{2\over
      2-\gamma} -
    \gamma\bigr) \\
    & = {(2-\gamma)^2 \over 8}\d_x^2 + \Bigl({\gamma(2-\gamma) \over 16}
    - {\gamma^2\over 32}\Bigr) x^{-2} - {x^{2\gamma-2\over
        2-\gamma}\over 2} \bigl(x^{2\over 2-\gamma} - \gamma\bigr)
  \end{equs}
  This is a Schr\"odinger operator with a potential that behaves at
  infinity like $x^{2\gamma/(2-\gamma)}$.  It follows that
  \begin{claim}
  \item If $1 > \gamma > 0$, then the operator $\CL_2$ has compact resolvent.
  \item If $\gamma = 0$, the operator $\CL_2$ does not have compact
    resolvent, but it still has a spectral gap (since one can see that
    its essential spectrum is the interval $[1/2,\infty)$).
  \item If $\gamma < 0$, then $0$ belongs to the essential spectrum of $\CL_2$.
  \end{claim}
  See for example \cite{RS78IV} for a proof.  It is then a natural
  conjecture that the spectrum of the generator of \eref{e:main} on the
  $L^2$ space weighted by the invariant measure has the same behavior
  (as a function of the parameter $\gamma = {2n \over k} + 1-2n$) as
  just described. The next section is a step towards a proof of this
  conjecture.

  \section{Lack of spectral gap}
  \label{sec:negative}

  The aim of this section is to obtain information on the location of
  the essential spectrum of the generator $\CL$ for \eref{e:main}. This
  will be accomplished by using ideas from averaging/homogenization
  theory to build a set of approximate eigenvectors. Since $\CL$ is not
  self-adjoint, there are various possible definitions of its essential
  spectrum (see \cite{EdmEvan} or \cite{GustWeid} for a survey). 
  We choose to retain the following definition:
  \begin{definition}
  For $T$ a closed densely defined operator on a Banach space $\CB$, the essential spectrum
  $\sigma_e(T)$ is defined as the set of all values $\lambda \in \C$ such that 
  $T - \lambda$ is not a semi-Fredholm operator.
  \end{definition}

  The set $\sigma_e(T)$ is contained in the corresponding sets for all other
  alternative definitions of the essential spectrum appearing in the abovementioned works.
  In this sense, the statement ``$\lambda \in \sigma_e(T)$'' used here is the strongest.
  In  particular, it is contained in the set 
  \begin{equ}
  \bigcap_{K \in \CK(\CB)} \sigma(T + K)\;,
  \end{equ}
  where $\CK(\CB)$ denotes the ideal of all compact operators on $\CB$ and $\sigma(T)$ denotes
  the spectrum of an operator $T$.

  We will make use of the following generalization of Weyl's criterion \cite{RS78IV}, 
  which gives a useful criterion for
  identifying the essential spectrum:

  \begin{proposition}\label{prop:compact}
    Let $T$ be a closed densely defined operator on a Banach space $\CB$.
   For any $\lambda\in\C$, $\lambda \in \sigma_e(T)$ if and only if there exist
    sequences $\phi_n$ and $\phi_n^*$ of elements in $\CD(T) \subset \CB$ and  $\CD(T^*) \subset \CB^*$ respectively with
    $\|\phi_n\| = \|\phi_n^*\| = 1$ and having no convergent subsequence such that $\lim_{n \to
      \infty} \|T\phi_n - \lambda \phi_n\| = \lim_{n \to
      \infty} \|T^*\phi_n^* - \bar \lambda \phi_n^*\| = 0$. 
  \end{proposition}

  \begin{proof}
  See \cite[Theorem~9.1.3]{EdmEvan} and \cite[Theorem~IV.5.11]{Kato80}, or the original work \cite{Wolf}.
  \end{proof}

  There are situations in which, even though it is difficult to locate the essential spectrum precisely,
  one can nevertheless exhibit a sequence $\phi_n$ as above such that $T\phi_n$ remains bounded.
  In that case, one has:

  \begin{proposition}\label{prop:compact2}
  Let $T$ be a closed densely defined operator on a Banach space $\CB$ with non-empty resolvent set.
  If there exists
  a  sequence $\phi_n$ of elements in $\CD(T) \subset \CB$ with
    $\|\phi_n\| = 1$ and having no convergent subsequence such that $\limsup_{n \to
      \infty} \|T\phi_n\| < \infty$, then $T$ does not have compact resolvent. 
  \end{proposition}
  \begin{proof}
  This claim follows from the fact that the
  compactness of the resolvent is equivalent to the statement that
  sets of the form $\{\phi\,|\, \|\phi\| \le 1\; \& \; \|T\phi\| \le K\}$ are precompact.
  \end{proof}

  \subsection{Scalings}
  \label{sec:scalings}

  Before we turn to the study of the generator, we make some remarks about
  the regularity and the behavior of functions that scale in a particular way.

  Denote by $X_{\Hf} = P\d_{Q} - Q |Q|^{2k-2} \d_{P}$ the Liouville
  operator associated to the ''free'' oscillator
  \begin{equ}[e:defHhat]
    \Hf(P,Q) = {P^2 \over 2} + {|Q|^{2k}\over 2k}\;.
  \end{equ}
  The constant $k$ is not necessarily an integer, so that in general the
  function $\Hf$ is not $\CC^\infty$ but only $\CC^{[2k]}$, where
  $[2k]$ is the integer part of $2k$.  We introduce the following
  definition:
  \begin{definition}\label{def:scaling}
    A function $\psi \colon \R^2\setminus\{0\} \to \R$ is said to
    \textit{scale like $\Hf^\alpha$} if it satisfies the relation
    $\psi(\lambda P, \lambda^{1 \over k} Q) = \lambda^{2\alpha} \psi(P,Q)$.
  \end{definition}

  One has the following elementary result:
  \begin{lemma}\label{lem:scaling}
    If $\psi \in \CC^1(\R^2\setminus\{0\})$ scales like $\Hf^\alpha$, then
    $\d_P \psi$ scales like $\Hf^{\alpha - {1\over 2}}$, $\d_Q \psi$ scales like
    $\Hf^{\alpha-{1\over 2k}}$, and $X_{\Hf}\psi$ scales like $\Hf^{\alpha + {1\over
        2} - {1\over 2k}}$.
  \end{lemma}

  Given some fixed $\alpha$, there is a one-to-one correspondence
  between functions $\phi$ that scale like ${\Hf}^\alpha$ and functions on
  the circle $S^1$ in the following way.  Define a function $r\colon S^1
  \to \R_+$ by the unique positive solution to
  \begin{equ} {r^2(\theta) \cos^2 \theta \over 2} + {r^{2k}(\theta)
      |\sin \theta|^{2k} \over 2k} = 1\;,
  \end{equ}
  and set
  \begin{equ}
    \bigl(\CS \phi\bigr)(\theta) = \phi\bigl(r(\theta)\cos\theta,
    r(\theta)\sin\theta\bigr)\;.
  \end{equ}
  The function $r$ is bounded away from $0$ (it is actually between
  $\sqrt 2$ and $(2k)^{1/(2k)}$) and, by the implicit functions theorem,
  it is of class $\CC^{[2k]}$. A straightforward but slightly tedious
  calculation shows that, via $\CS$, the operator $X_{\Hf}$ is conjugated to
  the differential operator
  \begin{equ}
    \bigl(\tilde X_{\Hf} f\bigr)(\theta) = \omega(\theta)
    f'(\theta)\;,\qquad \omega(\theta) = \cos^2\theta + r^{2k-2}(\theta)
    |\sin\theta|^{2k}\;.
  \end{equ}

  \begin{definition}
    A function $\psi \colon \R^2 \to \R$ that scales like ${\Hf}^\alpha$ is
    said to average out to $0$ if
  \begin{equ}
    \int_0^{2\pi} {\bigl(\CS \psi\bigr)(\theta) \over \omega(\theta)}
    d\theta = 0\;.
  \end{equ}
  \end{definition}
  With these preliminaries, the following result is now straightforward:

  \begin{proposition}\label{prop:scale}
    Let $\psi \in \CC^{r}(\R^2\setminus\{0\})$ scale like $\Hf^\alpha$ and
    average out to $0$.  Then, there exists a unique solution $\phi$ to
    the equation
  \begin{equ}[e:Poisson]
  X_{\Hf} \phi = \psi\;,
  \end{equ}
  such that $\phi$ also averages out to $0$. Furthermore, $\phi$ scales
  like $\Hf^{\alpha+{1\over 2k}-{1\over2}}$ and one has $\phi \in
  \CC^{r'}(\R^2\setminus\{0\})$ with $r' = \min\{[2k], r+1\}$.
  \end{proposition}

  \begin{proof}
    Let $\phi$ be the unique function scaling like ${\Hf}^{\alpha+{1\over
        2k}-{1\over2}}$ and such that
  \begin{equ}
    \bigl(\CS \phi\bigr)(\theta) = \int_0^\theta {\bigl(\CS
      \psi\bigr)(t) \over \omega(t)} dt - {1\over 2\pi} \int_0^{2\pi}
    \int_0^\theta {\bigl(\CS \psi\bigr)(t) \over \omega(t)} dt
    \,d\theta\;.
  \end{equ}
  One can check that one has indeed $X_{\Hf} \phi =\psi$. Furthermore, it
  follows from their explicit expressions that both $\CS$ and $\CS^{-1}$
  map $\CC^{r}$ functions into $\CC^{r}$ functions as long as $r \le
  [2k]$.
  \end{proof}

  We conclude this section with a small lemma that allows us to compute the
  $L^2$ norm of functions that scale in a certain way. Let $(P,Q) \in \R^2$ and
  $(x,y) \in \R^{2n}$ for some $n \ge 1$. The functions that will be considered in 
  the remainder of this section will always be of the form
  \begin{equ}[e:genericfcn]
  \phi_\CE(P,Q,x,y) = F(x - g(P,Q), y-h(P,Q)) \psi(P,Q) \chi(\Hf(P,Q)/\CE)\;,
  \end{equ}
  for some parameter $\CE > 0$. One has
  \begin{lemma}\label{lem:norm}
    Let $F \in L^2(\R^{2n})$, let $g,h \colon \R^2 \to \R^n$ be
    measurable, let $\chi\colon \R_+ \to \R$ be continuous and compactly
    supported away from $0$, and let $\psi$ be a function that is
    continuous away from $0$ and scales like $\Hf^\alpha$ for some
    $\alpha \in \R$.  Then, the
    $L^2$-norm of $\phi_\CE$ defined as in \eref{e:genericfcn} satisfies
    $\|\phi_\CE\| \propto \CE^{\alpha + {1\over 4} + {1\over 4k}}$.
  \end{lemma}

  \begin{proof}
    Make the change of variables $(\tilde P, \tilde Q) = (\CE^{1 \over 2} P,
    \CE^{1 \over 2k} Q)$ and use the scaling properties of $\psi$.
  \end{proof}

  \subsection{The case of three oscillators}

  Before we tackling  the general case of a chain with arbitrary length,
  let us ``cut our teeth'' on the problem with three oscillators.  Since
  we do not have an explicit expression for the invariant measure $\mu$
  (indeed, at this stage, we do not even know that it exists!), we are
  going to study the generator of \eref{e:main} in spaces of the type
  $L^2(e^{-\beta H}\,dp\,dq)$. As in \cite{EPR99NES,EckHai00NES}, it is
  not expected that the qualitative nature of the spectrum of $\CL$
  depends on the choice of $\beta$, as long as $\beta < 2\min\{\beta_0,
  \beta_2\}$ (as usual, we set $\beta_i = 1/T_i$).  Since one expects
  the true invariant measure to be somehow ''in between'' the Gibbs
  measures at temperatures $T_0$ and $T_2$, it is very likely that the
  qualitative nature of the spectrum of $\CL$ in $L^2(\mu)$ is also the
  same.


  We write
  \begin{equs}
    H_0 &= {p_0^2 \over 2} + {|q_0|^{2k} \over 2k} + {q_0^2 \over 2}
    \;,\qquad
    H_2 = {p_2^2 \over 2} + {|q_2|^{2k} \over 2k}  + {q_2^2 \over 2} \;,\\
    H_1 &= {p_1^2 \over 2} + {|q_1|^{2k} \over 2k} \;.
  \end{equs}
  \begin{remark}
    One should not think of $H_0$ (and $H_2$) as the energy one gets by
    pinning $q_1$ at $0$, but rather as the energy such that the
    corresponding force is the one that gets averaged out over the fast
    motion of the middle oscillator.
  \end{remark}

  With this notation, the Liouville operator $X_H = \d_p H \, \d_q -
  \d_q H\, \d_p$ for the total Hamiltonian can be broken up as follows:
  \begin{equ}
    X_H = X_{H_0} + X_{H_2} + X_{H_1} + (q_0 + q_2 - 2q_1)\,\d_{p_1} +
    q_1 \bigl(\d_{p_0} + \d_{p_2}\bigr)\;.
  \end{equ}
  Recall that the generator of the stochastic dynamics is given by
  \begin{equ}
    \CL = X_H - \gamma_0 p_0\d_{p_0} + \gamma_0 T_0 \d_{p_0}^2 -
    \gamma_2 p_2\d_{p_2} + \gamma_2 T_2 \d_{p_2}^2\;.
  \end{equ}
  The space $L^2(e^{-\beta H}\,dp\,dq)$ isometric to $L^2$ via the
  operator $Kf = e^{-\beta H/2} f$. This shows that $\CL$ is conjugate
  to the operator $\tilde \CL = K \CL K^{-1}$ on the flat $L^2$ space
  given by
  \begin{equs}
    \tilde \CL &= X_H + \sum_{i=0,2}\gamma_i T_i \Bigl( \bigl(\alpha_i -
    \alpha_i^*\bigr)p_i \d_{p_i} +
    \d_{p_i}^2 - \alpha_i \alpha_i^* p_i^2 + \alpha_i \Bigr) \\
    & = X_H + \gamma_0 T_0 \CL_\OU^0 + \gamma_2 T_2 \CL_\OU^2\;.
  \end{equs}
  where we set
  \begin{equ}
    \alpha_i = {\beta \over 2}\;,\qquad \alpha_i^* = {1\over T_i}
    -{\beta\over 2}\;.
  \end{equ}

  \begin{remark}
    Here we see the importance of the condition $\beta < 2\min\{\beta_0,
    \beta_2\}$: it makes sure that the coefficients in front of $p_i^2$
    are strictly negative. If this is not the case, $\tilde \CL$ is not
    dissipative anymore and does therefore not generate a
    $C_0$-semigroup on $L^2(e^{-\beta H}\, dp\, dq)$.
  \end{remark}

  The main result of this section is:

  \begin{theorem}\label{theo:threeosc}
    If $k \ge 2$, then the operator $\tilde \CL$ does not have compact
    resolvent for any $\beta < 2\min\{\beta_0,\beta_2\}$. If $k > 2$,
    then it has essential spectrum at $0$.
  \end{theorem}

  \begin{proof}
    The aim is to construct a sequence $\phi_n$ of \textit{approximate
    eigenfunctions} such that all the $\phi_n$ are mutually orthogonal,
    $\|\phi_n\| = 1$ and $\|\tilde \CL\phi_n\|$ either stays bounded or
    converges to $0$.  By Propositions~\ref{prop:compact} and \ref{prop:compact2}, this would then
    immediately imply a lack of compactness for the resolvent of $\tilde
    \CL$, or even the presence of essential spectrum at $0$. Since the
    spectral properties of $\tilde \CL$ and of its adjoint $\tilde
    \CL^*$ are the same, we can also construct such a sequence of
    approximate eigenfunctions for $\tilde \CL^*$ instead. They can then
    be interpreted as \textit{approximate invariant measures} for the dynamic
    \eref{e:main}.  Since it seems to be a little bit easier to get an
    intuition about densities of approximate invariant measures rather
    than about approximately invariant observables, this is what we are
    going to do in this section.

  The adjoint of $\tilde \CL$ is given by
  \begin{equs}
    \tilde \CL^* &= -X_H + \gamma_0 T_0 (\CL_\OU^0)^* + \gamma_2 T_2 (\CL_\OU^2)^* \\
    ( \CL_\OU^i)^* &= \bigl(\alpha_i^* - \alpha_i \bigr)p_i \d_{p_i} +
    \d_{p_i}^2 - \alpha_i \alpha_i^* p_i^2 + \alpha_i^*\;.
  \end{equs}
  Direct calculation shows that $\exp(- \alpha_i H_i)$ is an
  eigenfunction with eigenvalue $0$ for $\CL_\OU^i$ and $\exp(-
  \alpha_i^* H_i)$ is an eigenfunction with eigenvalue $0$ for
  $(\CL_\OU^i)^*$ with $i \in \{0,2\}$. However, if $T_0 \neq T_2$ (and
  therefore $\alpha_0 \neq \alpha_2$) it is not possible in general to
  find a closed expression for an eigenfunction with eigenvalue $0$ for
  $\tilde \CL^*$. Note also that $\alpha_i = \alpha_i^*$ if and only if
  $T_i = 1/\beta$, which is not surprising since in this case
  $\CL_\OU^i$ is self-adjoint.

  Choose now a function $\chi \colon \R \to [0,1]$ which is smooth and
  compactly supported on $[1,2]$ and set for example $\CE_n = 3^n$, so
  that the functions $\chi(H_1/\CE_n)$ have disjoint support.  The
  formal calculation performed in Section~\ref{sec:formalCalc} suggests
  that when the energy of the middle oscillator is large, the dynamic of
  \eref{e:main} keeps that energy approximately constant, while the two
  boundary oscillators equilibrate approximately at Gibbs measures at
  temperatures $T_0$ and $T_2$ respectively.  Our first guess would be
  therefore to build approximate eigenfunctions for $\tilde \CL$ and
  $\tilde \CL^*$ by setting
  \begin{equ}
    \phi_n = C_n e^{- \alpha_0 H_0 - \alpha_2 H_2}
    \chi(H_1/\CE_n)\;,\quad \phi_n^* = C_n^* e^{- \alpha_0^* H_0 -
      \alpha_2^* H_2} \chi(H_1/\CE_n)\;.
  \end{equ}
  Here, the constants $C_n$ and $C_n^*$ are chosen such that $\|\phi_n\|
  = \|\phi_n^*\| = 1$.  From \lem{lem:norm}, one infers that $C_n
  \propto C_n^* \propto \CE_n^{-{1\over 4k} - {1\over 4}}$. With this
  guess, we get
  \begin{equ}
    \tilde \CL\phi_n = C_ne^{-\alpha_0 H_0 - \alpha_2 H_2} \Bigl({p_1 (q_0 +
      q_2 - 2q_1)\over \CE_n} \chi'\Bigl({H_1 \over \CE_n}\Bigr) - q_1
    (\alpha_0 p_0 + \alpha_2 p_2) \chi\Bigl({H_1 \over
      \CE_n}\Bigr)\Bigr)\;,
  \end{equ}
  and similarly for $\tilde \CL^* \phi_n^*$. Since $p_1 \approx
  \CE_n^{{1\over 2}} \ll \CE_n$ (on the support of $\phi_n$), the first
  term goes to $0$ in $L^2$. The second term however goes to $\infty$
  because of the factor $q_1$, so we have to be a little bit more
  careful in our analysis.

  The problem is that we have not so far exploited the fact that we also
  know approximately what the fast oscillations superimposed over the
  \textit{slow} dynamic of the boundary oscillators look like, see
  \eref{e:osccontr}. These oscillations can be expressed as a function
  $\Phi$ of the middle oscillator, solution
  to the \textit{Poisson equation}
  \begin{equ}[e:Poisson1]
    X_{\Hf} \Phi(P,Q) = -Q\;.
  \end{equ}
  By \prop{prop:scale}, this equation has a unique solution that
  averages to $0$ along orbits of the solutions corresponding to $\Hf$.
  Furthermore, $\Phi$ scales like $\Hf^{{1\over k}-{1\over 2}}$.  In
  particular, we note that $\Phi$ is bounded when $k=2$ and converges to
  $0$ at large energies when $k > 2$.

  Our
  next guess is therefore to compensate for these fast oscillations by
  setting $\bar p_i = p_i + \Phi(p_1, q_1)$ and taking
  \begin{equ}
    \phi_n = C_n \exp\Bigl({- \alpha_0 H_0\bigl(\bar p_0, q_0\bigr) -
      \alpha_2 H_2\bigl(\bar p_2, q_2\bigr)}\Bigr) \chi\Bigl({H_1\over
      \CE_n}\Bigr)\;.
  \end{equ}
  Observe at this stage that $H_1 = \Hf(p_1, q_1)$, so that we can make use of \eref{e:Poisson1}
  when computing $X_{H_1} \phi_n$.
  One has, for $i = 0,2$:
  \begin{equs}
    X_{H_i} \phi_n &= -\alpha_i (q_i|q_i|^{2k-2} + q_i) \Phi \phi_n\;,\\
    X_{H_1} \phi_n &= \alpha_0 \bar p_0 q_1
    \phi_n
    + \alpha_2 \bar p_2 q_1 \phi_n  \;, \label{e:XH1phi}\\
    q_1 \bigl(\d_{p_0} + \d_{p_2}\bigr) \phi_n &= - \alpha_0 \bar p_0 q_1 \phi_n
    - \alpha_2 \bar p_2q_1 \phi_n\;,\label{e:crossphi}\\
    (q_0 + q_2 - 2q_1) \d_{p_1} \phi_n &= (q_0 + q_2 - 2q_1)\Bigl(-\bigl(\alpha_0 \bar p_0
    +\alpha_2 \bar p_2\bigr)\d_{P}\Phi
    + {p_1 \over \CE_n} {\chi'\over \chi}\Bigr)\phi_n\;,\\
    \CL_{\OU}^i \phi_n &=  \alpha_i \Phi \bigl((\alpha_i + \alpha_i^*) p_i + \alpha_i 
    \Phi \bigr) \phi_n\;.
  \end{equs}
  Here, we have omitted the argument $(p_1, q_1)$ of $\Phi$ and the 
  argument $H_1 / \CE$ of $\chi$ and $\chi'$ for the sake of simplicity.
  Note now that \eref{e:XH1phi} and \eref{e:crossphi} cancel each other out
  exactly. It follows from \lem{lem:norm} that 
  \begin{equs}
    \|X_{H_i} \phi_n\| &\les \CE_n^{{1\over k} - {1\over 2}}\;, \\
    \|(q_0 + q_2 - 2q_1) \d_{p_1} \phi_n\| &\les (1 + \CE_n^{1\over
      2k})\bigl(\CE_n^{{1\over k} - 1} + \CE_n^{-{1\over 2}}\bigr)
    \les \CE_n^{{1\over 2k} - {1\over 2}} \;,\\
    \|\CL_{\OU}^i \phi_n\| &\les \CE_n^{{1\over k} - {1\over 2}} +
    \CE_n^{{2\over k} - 1} \les \CE_n^{{1\over k} - {1\over 2}}\;.
  \end{equs}
  Note that the exponent ${1\over 4} + {1\over 4k}$ appearing in
  \lem{lem:norm} is precisely canceled by the normalization constant
  $C_n$.  We also used here the symbol $\Psi_1 \les \Psi_2$ for two
  expressions $\Psi_i$ as a shorthand for ``there exists a constant $C$
  such that $\Psi_1 \le C \Psi_2$.'' It follows from the above bounds
  that $\|\CL\phi_n\| \les \CE_n^{{1\over k} - {1\over 2}}$.

  It is possible to construct approximate eigenfunctions $\phi_n^*$ for
  $\CL^*$ similarly by setting
  \begin{equ}
    \phi_n^* = C_n \exp\Bigl({- \alpha_0^* H_0\bigl(\bar p_0, q_0\bigr)
      - \alpha_2^* H_2\bigl(\bar p_2, q_2\bigr)}\Bigr)
    \chi\Bigl({H_1\over \CE_n}\Bigr)\;.
  \end{equ}
  Note now that the only difference between $\CL$ and $\CL^*$ is that
  one changes the sign of $X_H$ and switches $\alpha_i$ and
  $\alpha_i^*$. This shows that the cancellation between \eref{e:XH1phi}
  and \eref{e:crossphi} still takes place when applying $\CL^*$ to
  $\phi_n^*$, so that $\|\CL^*\phi_n^*\| \les \CE_n^{{1\over k} -
    {1\over 2}}$ as above.

  If $k = 2$, it follows that there exists a constant $C$ such that
  $\|\CL\phi_n\| + \|\CL^*\phi_n^*\| \le C$ for every $n$. If $k > 2$,
  all the exponents appearing the the above expressions are negative, so
  that $\lim_{n \to \infty} (\|\CL\phi_n\|+ \|\CL^*\phi_n^*\|) = 0$.
  Applying Propositions~\ref{prop:compact} and \ref{prop:compact2} concludes the proof of the theorem.
  \end{proof}

  \begin{remark}
    It is clear from the proof that the exact same result also holds for
    a chain consisting of $4$ oscillators instead of $3$. One can
    construct approximate invariant measures in exactly the same way,
    but one has the additional freedom of choosing to take the energy of
    either of the two middle oscillators to be large.
  \end{remark}

  \subsection{Longer chain}

  In this section, we consider a chain of length $N+1$ for $N \ge 4$. 
  We will show that if $k > {3\over 2}$, then the generator of the dynamic has 
  essential spectrum at $0$. To this end,
  define similarly as before
  \begin{equs}
    H_0 &= {p_0^2 + p_1^2 \over 2} + {|q_0|^{2k} + |q_1|^{2k} \over 2k}
    +
    {(q_0-q_1)^2 + q_1^2 \over 2}\;, \\
    H_r &= {q_3^2 \over 2} + \sum_{i=3}^{N}\Bigl({p_i^2 \over 2} +
    {|q_i|^{2k}\over 2k}\Bigr)
    + \sum_{i=4}^{N} {(q_i - q_{i-1})^2 \over 2}\;, \\
    H_c &= {p_2^2 \over 2} + {|q_2|^{2k} \over 2}\;,
  \end{equs}
  so that
  \begin{equ}
    X_H = X_{H_0} + X_{H_r} + X_{H_c} + (q_1 + q_3 - 2q_2)\d_{p_2} + q_2
    \bigl(\d_{p_1} + \d_{p_3}\bigr)\;.
  \end{equ}
  As in the previous section, we consider the operator $\tilde \CL$ on the
  flat $L^2$ space given by
  \begin{equ}
  \tilde \CL = X_H + \gamma_0 T_0 \CL_\OU^0 + \gamma_{N} T_{N} \CL_\OU^r\;.
  \end{equ}
  We have

  \begin{theorem}\label{theo:compactchain}
    If $N+1 \ge 5$ and $k > {3\over 2}$, then the operator $\tilde \CL$ has
    essential spectrum at $0$.  If $k = {3\over 2}$, it does not have
    compact resolvent. As previously, these statements are independent
    of the value of $\beta < 2\min\{\beta_0, \beta_{N}\}$.
  \end{theorem}

  \begin{proof}
    Let $\Phi$ be defined as in the previous subsection. In this
    section, we do not only add a corrector term to $p_1$ and $p_3$, but
    also to $q_1$ and $q_3$.  We define $\Phi^{(2)}$ to be the solution to
    the Poisson equation
    $X_{\Hf} \Phi^{(2)} = \Phi$ and we define new variables $\bar p$ and $\bar
    q$ by
  \begin{equ}
    \bar p_{i} = p_{i} + \Phi(p_2,q_2)\;,\quad \bar q_{i} =
    q_{i}  + \Phi^{(2)}(p_2,q_2)\;,
  \end{equ}
  for $i = 1,3$ and $(\bar p_i, \bar q_i) = (p_i,q_i)$ otherwise. With
  this notation, we set as before
  \begin{equs}
    \phi_n &= C_n \exp\bigl({- \alpha_0 H_0(\bar p, \bar q) - \alpha_{N}
      H_r(\bar p, \bar q)}\bigr) \chi\bigl(H_c/\CE_n\bigr)\;.
  \end{equs}

  In order to compute $X_H \phi_n$, let us first apply $X_H$ to $\bar p_i$ and $\bar q_i$:
  \begin{equs}
    X_H \bar q_0 &= \bar p_0\;,\\
    X_H \bar p_0 &= - \bar q_0 |\bar q_0|^{2k-2} + \bar q_1 - \bar q_0 - \Phi^{(2)}\;,\\
    X_H \bar q_1 &= \bar p_1 + (q_0 + q_2 - 2 q_1) \d_P\Phi^{(2)}\;,\\
    X_H \bar p_1&=  - \bar q_1|\bar q_1|^{2k-2}  + \bar q_0 - 2\bar q_1 \\
    &\quad + \bigl(\bar q_1|\bar q_1|^{2k-2} - q_1 |q_1|^{2k-2}\bigr) +
    (q_0 + q_2 - 2q_1) \d_P \Phi + 2\Phi^{(2)}\;.
  \end{equs}
  Hence
  \begin{equs}
    X_H H_0(\bar p, \bar q) &= - \bar p_0 \Phi^{(2)} - \bigl(\bar q_1
    |\bar q_1|^{2k-1} + 2\bar q_1 - \bar q_0\bigr)
    \bigl(q_0 + q_2 - 2q_1\bigr) \d_{P} \Phi^{(2)} \label{e:XHH0}\\
    &\quad + \bar p_1 \bigl(\bar q_1|\bar q_1|^{2k-2} - q_1
    |q_1|^{2k-2}\bigr) + \bar p_1 \bigl(q_0 + q_2 - 2q_1\bigr) \d_{P}
    \Phi
    + 2 \bar p_1 \Phi^{(2)} \\
    &= R_1 + R_2 + R_3 + R_4 + R_5\;,
  \end{equs}
  and similarly for $X_H H_r(\bar p, \bar q)$ by symmetry. We have
  furthermore
  \begin{equ}
  X_H H_c =p_2 (q_1 + q_3 - 2q_2)\;,
  \end{equ}
  so that 
  \begin{equ}
    X_H \chi(H_c / \CE_n) = {q_1 + q_3 - 2q_2 \over \CE_n} p_2 \chi'(H_c
    / \CE_n)\;.
  \end{equ}
  Since $\CL_\OU^0 \phi_n = \CL_\OU^r \phi_n= 0$, we thus have
  \begin{equ}
    \tilde \CL \phi_n = \Bigl(- \alpha_0 X_H H_0(\bar p, \bar q) -
    \alpha_N X_H H_r(\bar p, \bar q) + {X_H \chi(H_c / \CE_n) \over
      \chi(H_c/\CE_n)} \Bigr)\phi_n\;.
  \end{equ}
  We bound the terms appearing in this expression in the same way as in
  the previous subsection.  Since $\Phi$ scales like $\Hf^{{1\over k} -
    {1\over 2}}$ and $\Phi^{(2)}$ scales like $\Hf^{{3\over 2k} - 1}$,
  \lem{lem:norm} shows that
  \begin{equs}[2]
    \|R_1 \phi_n\| &\les \CE_n^{{3\over 2k}-1} \;,\quad& \|R_2 \phi_n\|
    &\les \bigl(1 + \CE_n^{1\over 2k}\bigr) \CE_n^{{3\over 2k}-{3\over
        2}} \les \CE_n^{{2\over k} - {3\over
        2}}\;,\\
    \|R_3 \phi_n\| &\les \CE_n^{{3\over 2k}-1} \;,\quad& \|R_4 \phi_n\|
    &\les \bigl(1 + \CE_n^{1\over 2k}\bigr)
    \CE_n^{{1\over k}-1} \les \CE_n^{{3\over 2k} - 1}\;,\\
    \|R_5 \phi_n\| &\les \CE_n^{{3\over 2k}-1} \;,\quad&
    \bigl\|\textstyle{X_H \chi \over \chi} \phi_n\bigr\| &\les \bigl(1 +
    \CE_n^{1\over 2k}\bigr) \CE_n^{-{1\over 2}} \les \CE_n^{{1\over 2k}
      - {1\over 2}}\;.
  \end{equs}
  Collecting all these bounds, we obtain $\|\tilde \CL \phi_n\| \les
  \max \bigl\{\CE_n^{{3\over 2k} - 1}, \CE_n^{{1\over 2k} - {1\over 2}}
  \bigr\}$.  As before, if we set
  \begin{equs}
    \phi_n^* &= C_n \exp\bigl({- \alpha_0^* H_0(\bar p, \bar q) -
      \alpha_{N}^* H_r(\bar p, \bar q)}\bigr)
    \chi\bigl(H_c/\CE_n\bigr)\;,
  \end{equs}
  we obtain the same bounds for $\|\tilde \CL^* \phi_n^*\|$.  The
  exponents appearing in all of these bounds are strictly negative
  whenever $k > 3/2$, thus concluding the proof of
  Theorem~\ref{theo:compactchain}.
  \end{proof}

  \section{Effective dynamics}
  \label{sec:effective}

  From now on, we study the case of three oscillators in detail. In this
  section, we derive an effective dynamic for the outer oscillators that
  is valid in the regime where most of the energy is located in the
  center oscillator.  More precisely, we show that there exists a change
  of variable $(p_i,q_i) \mapsto (\bar p_i, \bar q_i)$ for $i = 0,2$
  such that the equations of motion for $(\bar p_i, \bar q_i)$ decouple
  (to leading order) from the rest of the system, provided that the
  energy of the middle oscillator is large.

  As before, we will use throughout this section the symbol $\Psi_1 \les
  \Psi_2$ for two expressions $\Psi_i$ as a shorthand for ``there exists
  a constant $C$ such that $\Psi_1 \le C \Psi_2$.'' The constant $C$
  depends in general on the parameters of the model, but is of course
  independent of the arguments of the $\Psi_i$.

  \begin{theorem}\label{theo:decouple}
    Assume that $k > {3\over 2}$.  There exist smooth functions
    $\Phi_p^i$ and $\Phi_q^i$ depending on $(p_i, q_i, p_1, q_1)$ such
    that if we make the change of variables $\bar p_i = p_i + \Phi_p^i$
    and $\bar q_i = q_i + \Phi_q^i$ (for $i = 0,2$), the equations of
    motion for $(\bar p_i, \bar q_i)$ are given by
    \begin{equs}[e:motionbar]
      d \bar q_i &= \bar p_i\, dt + R_q^i\, dt + \Sigma_q^i\, dw_i \\
      d \bar p_i &= - \bar q_i |\bar q_i|^{2k-2}\, dt - \bar q_i\, dt -
      \gamma_i \bar p_i\, dt + R_p^i\, dt + \Sigma_p^i\, dw_i
    \end{equs}
    for some adapted processes $R_p^i$, $R_q^i$, $\Sigma_p^i$ and $\Sigma_q^i$.
    Furthermore, the \textit{error terms} $R$ and $\Sigma$ satisfy the bounds
    \begin{equs}[2][e:errorbounds]
      |R_p^i| &\les \bigl(\bar E_0 + \bar E_2)^{{1 \over 2} - \delta}
      &\qquad |R_q^i| &\les \bar E_i^{{1 \over 2k} - \delta} + {\bar
        E_{2-i}^{{1 \over 2k}}
        \over \bar E_{i}^{\delta}}\\
      |\Sigma_p^i| &\les 1 &\qquad |\Sigma_q^i| &\les \bar E_i^{-{1
          \over 2}}
    \end{equs}
    for some $\delta > 0$. Here, the energies $\bar E_i$ 
    are given by $\bar E_i = 1 + \Hf(\bar p_i, \bar q_i)$. 
  \end{theorem}

  \begin{proof}
    This theorem gives a change of variables where the high-speed
    oscillations due to the presence of a breather located on the middle
    oscillator (that is the case where $E_0, E_2 \ll E_1$) have been
    decoupled from the remaining degrees of freedom, leaving an
    effective ''averaged out'' dynamic.  Recalling the formal calculation
    performed in Section~\ref{sec:formalCalc}, we see that when the
    energy $E$ is predominantly concentrated in the central oscillator,
    then the amplitudes of the oscillations for $p_i$ and $q_i$, $i =
    0,2$ scale to leading order like $E^{{1 \over k} - {1 \over 2}}$ and
    $E^{{3\over 2k} - 1}$ respectively. This indicates that there are
    natural breakpoints at $k = 2$ and $k = 3/2$. When $k \ge 2$, the
    oscillations of both the $p_i$ and the $q_i$ are bounded as $E$
    increases, so that they can be removed by a change of variables
    which is a bounded perturbation of the identity.

    When $k < 2$, the amplitude of the oscillations of the $p_i$
    increases with $E$, but as long as $k \ge 3/2$, the amplitude of the
    $q_i$ does not.  This growth will cause extra difficulties. If we
    consider $k < 3/2$, both amplitudes would grow with $E$, leading to
    further complications. Since our goal is to outline the ideas
    without seeking the greatest generality, we resist the temptation to
    analyze all cases and restrict ourselves to the case $k > 3/2$.

    Before we proceed, let us compute the expressions $R$ and $\Sigma$
    for a generic choice of $\Phi_p^i$ and $\Phi_q^i$. Applying It\^os
    formula to $\bar p_i$ and $\bar q_i$, we obtain
  \begin{equs}[e:exprRS]
    R_p^i &= V'(q_i + \Phi_q^i) - V'(q_i) + \Phi_q^i + \gamma_i \Phi_p^i
    + \CL \Phi_p^i + q_1\\
    R_q^i &= \CL \Phi_q^i - \Phi_p^i \;,\quad \Sigma_p^i = \sigma_i
    \bigl(1+ \d_{p_i} \Phi_p^i\bigr)\;,\quad \Sigma_q^i = \sigma_i \,
    \d_{p_i} \Phi_q^i \;.
  \end{equs}
  Here and in the sequel we write $V(q)$ as a shorthand for
  ${|q|^{2k}\over 2k}$ and $\sigma_i$ as a shorthand for
  $\sqrt{2\gamma_i T_i}$.

  \textit{The case ${k \ge 2}$.}  The only ``bad'' term in the equations
  of motion for $(p_i, q_i)$ is the $q_1$ in the right hand side of the
  equation for $p_i$. The case $k \ge 2$ is much easier than the case $k
  < 2$ since the system is more ``rigid'' in the later case. One can then
  simply take $\Phi_q^i = 0$ and $\Phi_p^i = \Phi(p_1, q_1)$, where
  $\Phi$ is the centered solution to the Poisson equation
  \begin{equ}[e:defPhi]
  X_{\Hf} \Phi = \CR(P,Q) - Q
  \end{equ}
  and $\CR(P,Q) = Q \psi(\Hf(P,Q))$, where $\psi \colon \R \to [0,1]$ is a smooth function such that
  $\psi(x) = 1$ for $|x| \le 1$ and $\psi(x) = 0$ for $|x| \ge 2$.  Making
  this choice of $\Phi_p^i$ and $\Phi_q^i$ in \eref{e:exprRS} yields
    \begin{equs}[2]
      \Sigma_q^i &= 0\;, &\qquad R_q^i &=  -\Phi(p_1, q_1)\;,\\
      \Sigma_p^i &= \sigma_i\;, & R_p^i &= \gamma_i \Phi(p_1, q_1) +
      \d_P \Phi(p_1, q_1) \bigl(q_0 + q_2 - 2 q_1\bigr) + \CR(p_1,q_1)
    \end{equs}
    Since the function $\Phi$ scales like $\Hf^{{1\over k} - {1\over
        2}}$ outside of a compact set, it can be checked easily that the
    bounds \eref{e:errorbounds} hold, provided that $k \ge 2$.  Note
    that in this case, $E_i$ and $\bar E_i$ are equivalent in the sense
    that $E_1 \les \bar E_i \les E_i$ since $\Phi$ is bounded. Therefore
    all occurrences of $E_j$ in the bounds can be replaced by $\bar E_j$
    at the cost of multiplicative constants.

    \textit{The case ${3/2 < k < 2}$.}  We are now going to assume that
    $k < 2$, which the more delicate case. Note that the second and
    third terms in $R_p^i$ above satisfy the bounds \eref{e:errorbounds}
    (with $E_j$ instead of $\bar E_j$, but this problem will be dealt
    with later) provided that $Q\d_P \Phi$ scales like $\hat H^\theta$
    for some $\theta \le 0$. This is the case when $k \ge {3\over 2}$,
    which is one of the reasons why we restrict ourselves to this case.
    Therefore, only the terms involving $\Phi$ scale worse then the
    desired bounds on the error terms and need to be eliminated.  This
    motivates the introduction of the solution $\Phi^{(2)}$ to the
    Poisson equation
    \begin{equ}[e:defPhi2]
    X_{\Hf}
    \Phi^{(2)} = \Phi\;.
    \end{equ}
    Note that $\Phi^{(2)}$ scales like $\hat H^{{3\over 2k} - 1}$
    outside of a compact set, so that it is bounded if $k \ge {3\over
      2}$.  It would be tempting at this point to simply subtract
    $\gamma_i \Phi^{(2)}(p_1, q_1)$ to $p_i$ and add $\Phi^{(2)}(p_1,
    q_1)$ to $q_i$.

    This would however introduce correction terms that grow faster than
    the bounds in \eref{e:errorbounds}. The trick is to realize that
    these correction terms are multiplied with terms that go to $0$ as
    the energy of the middle oscillator becomes large.  We therefore
    multiply $\Phi^{(2)}$ with a cutoff function that makes sure that
    this second round of correction is applied only when the energy of
    the middle oscillator is large compared to the energy of the
    boundary oscillators.

    Let $\phi \colon \R_+ \to [0,1]$ be a smooth increasing function
    such that $\phi(x) = 0$ for $x \le 1$ and $\phi(x) = 1$ for $x \ge
    2$. Let furthermore $\alpha$ be an exponent to be determined later
    and set \minilab{e:Phi}
    \begin{equs}
      \Phi_p^i &= \Phi(p_1, q_1) - \gamma_i \Phi^{(2)}(p_1, q_1)\;,
      \label{e:Phip}\\
      \Phi_q^i &= \phi(E_1 / E_i^{\alpha}) \Phi^{(2)}(p_1, q_1)\;,
      \label{e:Phiq}
    \end{equs}
    where we defined $E_i = 1+\Hf(p_i, q_i)$.  In the sequel, we are
    going to use the shorthand $\phi_i^\alpha = \phi(E_1 / E_i^\alpha)$
    and we will omit the arguments of $\Phi$, $\Phi^{(2)}$ and
    $\phi_i^\alpha$ in order to simplify notations.  Before we turn to
    the verification of the bounds \eref{e:errorbounds}, we remark that
    since $k < 2$, $E_i$ and $\bar E_i$ are \textit{not} equivalent for
    $i = 0,2$. Since $k \ge 3/2$, $E_1$ and $\bar E_1$ are however
    equivalent in the sense that $E_1 \les \bar E_1 \les E_1$. Since we
    wish to bound the remainder in terms of powers of $\bar E_i$ and not
    $E_i$, we are now going to show how these quantities are related.
    From the definitions of $\bar p_i$ and $\bar q_i$, we have for $i =
    0,2$ the estimate
    \begin{equ}
      |\bar E_i - E_i| \les 1 + |p_i \Phi| \les 1+ \eps p_i^2 + {1\over
        \eps}\Phi^2 \les 1 + \eps E_i + {1\over \eps} E_1^{{2\over k} -
        1}\;.
  \end{equ}
  By choosing $\eps$ sufficiently small and moving the term $\eps E_i$
  to the left hand side, we thus obtain the two bounds
  \begin{equ}[e:boundsEEbar]
    E_i \les \bar E_i + E_1^{{2\over k} - 1}\;,\qquad \bar E_i \les E_i
    + E_1^{{2\over k} - 1}\;.
  \end{equ}

  Writing $g = E_1 / E_i^\alpha$, and applying It\^o's formula to $g$, we obtain
    \begin{equs}
      dg &= E_i^{-\alpha} p_1 (q_0 + q_2 - 2q_1)\, dt - \alpha E_1
      E_i^{-\alpha - 1} \bigl({\sigma_i^2\over 2}- \gamma_i p_i^2 + q_1
      - q_i\bigr)\, dt\\
      &\qquad + {\alpha (\alpha + 1) \over 2} E_1 E_i^{-\alpha -
        2}\sigma_i^2 p_i^2\, dt - \alpha E_1 E_i^{-\alpha-1} \sigma_i
      p_i\, dw_i\\
      &= I_1\, dt + I_2 \, dt + I_3\, dt + dM(t)\;.
    \end{equs}
    $M$ is a martingale since, on any finite time interval, it is easy
    to get control over the expected value of any power of the total
    energy. (Even a bound which grows exponentially with the length of
    the time interval is sufficient.)  Bounding each of the three terms
    $I_j$ separately we obtain:
  \begin{equs}
    |I_1| &\les E_i^{-\alpha + {1 \over 2k}}E_1^{1\over 2} +
    E_i^{-\alpha}E_{2-i}^{1\over 2k} E_1^{{1 \over 2}}
    + E_i^{-\alpha} E_1^{{1\over 2} + {1 \over 2k}}\;, \\
    |I_2| &\les E_1 E_i^{-\alpha} + E_1^{1+ {1\over 2k}} E_i^{-\alpha-1}\;, \\
    |I_3| &\les E_1 E_i^{-\alpha-1}\;.
  \end{equs}

    Applying It\^o's formula to $\phi_i^\alpha$, we get
    \begin{equs}
      \CL \phi_i^\alpha &= \phi'(E_1/E_i^\alpha)\, \CL g +
      \alpha^2\sigma_i^2 \phi''(E_1/E_i^\alpha) E_1^2 E_i^{-2\alpha -2}
      p_i^2 = I_4 + I_5\;.
    \end{equs}
    Note now that $\phi'$ and $\phi''$ are zero outside the interval
    $[1,2]$. Therefore, we have $E_1 \les E_i^\alpha \les E_1$ on the
    support of these functions.  This line of reasoning yields for $1/k
    \le \alpha \le 2k$ the bounds:
  \begin{equs}
    |I_4| &\les E_i^{{1\over 2}({1 \over k} - \alpha)} +
    E_i^{-{\alpha\over 2}} E_{2-i}^{{1 \over 2k}} + E_i^{{\alpha \over
        2}({1 \over k}-1)} + 1 + E_i^{{\alpha \over 2k} - 1} + E_i^{-1}
    \les  E_{2-i}^{{1 \over 2k}} \;,\\
    |I_5| &\les E_i^{-1} \les 1\;.
  \end{equs}
  Collecting these estimates and taking into account the support of
  $\phi'$ and $\phi''$ produces $|\CL \phi_i^\alpha| \les \bigl(1 +
  E_{2-i}^{{1 \over 2k}}\bigr)\one_{2E_i^\alpha \ge E_1 \ge
    E_i^\alpha}$.

  Recall that, from the scalings given in Lemma~\ref{lem:scaling} of the
  solutions to the Poisson equation \eref{e:Poisson} and the definitions
  \eref{e:defPhi} and \eref{e:defPhi2} of $\Phi$ and $\Phi^{(2)}$, we
  have that $|\Phi| \les E_1^{{1 \over k} - {1 \over 2}}$, $|\d_P\Phi|
  \les E_1^{{1 \over k} - 1}$, $|\Phi^{(2)}| \les E_1^{{3 \over 2k} -
    1}$, $|\d_P\Phi^{(2)}| \les E_1^{{3 \over 2k} - {3 \over 2}}$.  We
  will also use the fact that, from the definition of $\phi_i^\alpha$
  combined with \eref{e:boundsEEbar}, one has both bounds
  \begin{equ}[e:suppphi]
  E_i^\alpha \les E_1\;,\qquad \bar E_i^\alpha \les E_1\;,
  \end{equ}
  on the support of $\phi_i^\alpha$. Furthermore, one has the bound $E_1
  \les E_i^\alpha$ on the support of $1- \phi_i^\alpha$ which, by
  \eref{e:boundsEEbar}, implies that one also has $E_1 \les \bar
  E_i^\alpha$ on the support of $1- \phi_i^\alpha$.

  Using the definitions of $\Phi_p^i$ and $\Phi_q^i$ given in
  \eref{e:Phi}, equation \eref{e:exprRS} yields
  \begin{equ}
    R^i_{q} = - \Phi + \phi_i^{\alpha} \CL \Phi^{(2)} + (\CL
    \phi_i^{\alpha}) \Phi^{(2)} + \gamma_i \Phi^{(2)}\;.
  \end{equ}
  We now make use of the definitions of $\CL$ and $\Phi^{(2)}$ to obtain
  \begin{equ}
    \CL \Phi^{(2)} = \Phi + \bigl(q_0 + q_2 - 2q_1\bigr) \d_P \Phi^{(2)}
    \;.
  \end{equ}
  This allows us to obtain the following bounds for $R^i_{q}$:
  \begin{equs}
    |R^i_{q}| &= |(\phi_i^{\alpha}-1)\Phi + \phi_i^{\alpha} (q_0 +
    q_2-2q_1)\, \d_P\Phi^{(2)}
    + (\CL \phi_i^{\alpha}) \Phi^{(2)} + \gamma_i \Phi^{(2)}|\;, \\
    &\les (1 - \phi_i^{\alpha}) E_1^{{1 \over k} - {1 \over 2}} +
    \phi_i^\alpha E_1^{-{3\over 2} (1 - {1 \over k})} \bigl(E_i^{{1
        \over 2k}}+E_1^{{1 \over 2k}}+E_{2-i}^{{1 \over 2k}}\bigr) \\
    &\quad + E_{2-i}^{{1 \over 2k}} E_1^{{3 \over 2k} -
      1}\one_{2E_i^\alpha \ge E_1 \ge E_i^\alpha} + 1\\
    &= I_6 + I_7 + I_8 + 1\;.
  \end{equs}
  Our aim is to bound the terms $I_6$, $I_7$, $I_8$ in terms of the
  $\bar E_j$'s instead of the $E_j$'s.  To do so, we now fix
  \begin{equ}
  \alpha = {3\over 2}\;.
  \end{equ}
  Since $k \in (3/2, 2)$, $\alpha \in (1/k, 2k)$ which is the constraint
  that we had to impose earlier.

  Since, as mentioned above, $E_1 \les \bar E_i^\alpha$ on the support
  of $1- \phi_i^\alpha$, one has
  \begin{equ}
  |I_6| \les \bar E_i^{{\alpha\over k} - {\alpha\over 2}}\;.
  \end{equ}
  One can check that the choice $\alpha = 3/2$ implies that there exists
  $\delta > 0$ so that ${\alpha\over k} - {\alpha\over 2} \le {1 \over
    2k} - \delta$ (actually, one can set $\delta = 1/12$ for the range
  of values of $k$ considered here). This shows that $|I_6| \les \bar
  E_i^{{1 \over 2k} - \delta}$ as required.

  To bound $I_7$, we make use of the fact that
  \begin{equ}[e:boundsoumEi]
    E_i^{{1 \over 2k}} + E_{2-i}^{{1 \over 2k}} \les \bar E_i^{{1 \over
        2k}} + E_1^{{1\over k^2}-{1 \over 2k}} + \bar E_{2-i}^{{1 \over
        2k}} \les \bar E_i^{{1 \over 2k}} + E_1^{{1 \over 2k}} + \bar
    E_{2-i}^{{1 \over 2k}}\;,
  \end{equ}
  so that, by virtue of \eref{e:suppphi},
  \begin{equs}
    |I_7| &\les \phi_i^\alpha E_1^{-{3\over 2} (1 - {1\over k})}
    \bigl(\bar E_i^{{1 \over 2k}}+ \bar E_{2-i}^{{1 \over 2k}}\bigr) +
    \phi_i^\alpha E_1^{{2 \over k} - {3 \over 2}}\\
    &\les \bar E_i^{-\delta} \bigl(\bar E_i^{{1 \over 2k}}+ \bar
    E_{2-i}^{{1 \over 2k}}\bigr) + 1\;.
  \end{equs}

  We finally turn to $I_8$. We first remark that, by
  \eref{e:boundsEEbar} and the fact that $\alpha \bigl({2 \over k} -
  1\bigr) < 1$, one has $\bar E_i^\alpha \les E_1 \les \bar E_i^\alpha$
  on the support of the indicator function $\one_A$ where $A =
  \{{2E_i^\alpha \ge E_1 \ge E_i^\alpha}\}$. We thus obtain
  \begin{equs}
    |I_8| &\les E_{2-i}^{{1 \over 2k}} \bar E_i^{{3\alpha \over 2k} -
      \alpha} \one_A \les \bigl(\bar E_{2-i}^{{1 \over 2k}} +
    E_1^{{1\over 2k}({2\over k} - 1)}\bigr) \bar E_i^{{3\alpha \over 2k}
      - \alpha} \one_A \\ 
    &\les \bigl(\bar E_{2-i}^{{1 \over 2k}} + \bar E_i^{{\alpha\over
        2k}({2\over k} - 1)}\bigr) \bar E_i^{{3\alpha \over 2k} -
      \alpha} \les \bar E_{2-i}^{{1 \over 2k}} \bar E_i^{-\delta} + \bar
    E_i^{\alpha {1 + k - k^2 \over k^2}} \;.
  \end{equs}
  One can check that, for $\alpha = 3/2$ and the range of $k$'s of
  interest, the exponent of the last term is strictly smaller than
  $1/2k$. Therefore, $|I_8| \les \bar E_i^{-\delta} \bigl(\bar E_i^{{1
      \over 2k}}+ \bar E_{2-i}^{{1 \over 2k}}\bigr)$ as required,
  choosing $\delta$ smaller if necessary.  Collecting all of these
  bounds shows that $R_q^i$ does indeed satisfy the bound advertised in
  \eref{e:errorbounds}.

  Turning to $R_p^i$, \eref{e:exprRS} yields
  \begin{equs}
    |R^i_{p}| &= \big|\CR(q_1) + (q_0 + q_2-2q_1)\,\bigl( \d_P\Phi -
    \gamma_i \d_P\Phi^{(2)}\bigr) - (\gamma_i^2-\phi_i^{\alpha})
    \Phi^{(2)} \\
    &\qquad +
    V'\big(q_i + \phi_i^{\alpha} \Phi^{(2)}\big) -V'\big(q_i\big)\big|\;, \\
    & \les 1 + \bigl(\bar E_0^{{1 \over 2k}} + \bar E_2^{{1 \over 2k}} +
    E_1^{{1 \over 2k}}\bigr)E_1^{{1 \over k}-1} + 1+ \bigl|V'\big(q_i +
    \phi_i^{\alpha} \Phi^{(2)}\big) -V'\big(q_i\big)\bigr|\;.
  \end{equs}
  Here, we made use of \eref{e:boundsoumEi} and of \eref{e:boundsEEbar}.
  Since, for $k \ge 3/2$, $\Phi^{(2)}$ is a bounded function, one has the further bound
  \begin{equs}
  \bigl|V'\big(q_i + \phi_i^{\alpha} \Phi^{(2)}\big) -V'\big(q_i\big)\bigr| &=
  \bigl|V'\big(\bar q_i\big) -V'\big(\bar q_i - \phi_i^{\alpha} \Phi^{(2)}\big)\bigr| \\
  & \les (1+|\bar q_i|)^{2k-2}  |\Phi^{(2)}|
  \les \bar E_i^{1-{1 \over k}}  \les \bar E_i^{{1 \over 2} - \delta}\;,
  \end{equs}
  for $\delta$ sufficiently small.
  Collecting these bounds and using the fact that $k > 3/2$, we obtain
  \begin{equ}[E;boundRpi]
    |R^i_{p}| \les \bar E_0^{{1 \over 2k}} + \bar E_2^{{1 \over 2k}} +
    \bar E_i^{{1 \over 2} - \delta} \les \bar E_i^{{1 \over 2} - \delta}
    + \bar E_{2-i}^{{1 \over 2} - \delta}\;,
  \end{equ}
  which is indeed of the form \eref{e:errorbounds}. 

  Lastly, the $\Sigma$-terms can be bounded by
  \begin{equs}
      |\Sigma_q^i| &= |\alpha \sigma_i \phi'(E_1 / E_i^\alpha) \Phi^{(2)} E_1
      E_i^{-\alpha-1} p_i| \les \bar E_i^{-{1 \over 2}} \;,\\ 
      \Sigma_p^i &= \sigma_i \;,
  \end{equs}
  where we made use of the fact that, like in the bound of $I_8$, $E_1$,
  $E_1^\alpha$ and $\bar E_i^\alpha$ are equivalent on the support of
  $\phi'(E_1 / E_i^\alpha)$. This concludes the proof of
  Theorem~\ref{theo:decouple}.
  \end{proof}

  In the case $k>2$, it will be useful in the sequel to have a better
  approximation of the dynamics that yields smaller error terms in the
  regime where most of the energy is located in the middle oscillator:

  \begin{theorem}\label{theo:effective2}
    Assume that $k > 2$.  There exist smooth functions $\Phi_p^i$ and
    $\Phi_q^i$ depending on $(p_i, q_i, p_1, q_1)$ such that if we make
    the change of variables $\bar p_i = p_i + \Phi_p^i$ and $\bar q_i =
    q_i + \Phi_q^i$ (for $i = 0,2$), the equations of motion for $(\bar
    p_i, \bar q_i)$ are given by
    \begin{equs}[e:motionbar2]
      d \bar q_i &= \bar p_i\, dt + R_q^i\, dt  \\
      d \bar p_i &= - \bar q_i |\bar q_i|^{2k-2}\, dt - \bar q_i\, dt -
      \gamma_i \bar p_i\, dt + R_p^i\, dt + \sigma_i\, dw_i
    \end{equs}
    for some adapted processes $R_p^i$ and $R_q^i$.  
    Furthermore, the \textit{error terms} $R$ satisfy the bounds
    \begin{equ}[e:errorbounds2]
      |R_p^i| \les \bigl(E_0 + E_2)^2 H^{{3 \over 2k}-1} \qquad
      |R_q^i| \les \bigl(E_0 + E_2) H^{{3 \over 2k}-1}\;. 
    \end{equ}
    Here, the energies $E_i$ are given as before by $E_i = 1 + \Hf(\bar
    p_i, \bar q_i)$ and $H$ is the total Hamiltonian of our system.
  \end{theorem}

  \begin{proof}
  Following the proof of Theorem~\ref{theo:decouple}, we set as in \eref{e:Phi}
  \begin{equs}
      \Phi_p^i &= \Phi(p_1, q_1) - \gamma_i \Phi^{(2)}(p_1, q_1)\;,
      \\
      \Phi_q^i &= \Phi^{(2)}(p_1, q_1)\;.
    \end{equs}
  This yields for $R_p^i$ and $R_q^i$ the expressions
  \begin{equs}
      R^i_{q} &= (q_0 + q_2-2q_1) \d_P \Phi^{(2)} + \gamma_i \Phi^{(2)}\;,\\
      R_p^i &= \CR(q_1) + (q_0 + q_2-2q_1)\,\bigl( \d_P\Phi -
      \gamma_i \d_P\Phi^{(2)}\bigr) - (\gamma_i^2-1)
      \Phi^{(2)} \\
      &\qquad +
      V'\big(q_i + \Phi^{(2)}\big) -V'\big(q_i\big)\;.
  \end{equs}
  The desired bounds now follow from Lemma~\ref{lem:trivial} below,
  together with the fact that both $\Phi^{(2)}$ and $Q\d_P \Phi$ scale
  like $\Hf^{{3 \over 2k}-1}$.
  \end{proof}

  \begin{lemma}\label{lem:trivial}
  For every $\alpha > 0$, the bound
  \begin{equ}
  x^{-\alpha} \le  (x+y)^{-\alpha} \max\{2^\alpha, y^\alpha\}
  \end{equ}
  holds for every $x, y \ge 1$.
  \end{lemma}

  \begin{proof}
    If $x \ge y$, then $x^{-\alpha} \le 2^\alpha (2x)^{-\alpha} \le
    2^\alpha (x+y)^{-\alpha}$.  If on the other hand $x \le y$, then
    $x^{-\alpha} \le 1 \le y^\alpha (x+y)^{-\alpha}$.
  \end{proof}

  \section{Existence of an invariant measure}
  \label{sec:existence}

  \subsection{General strategy}

  To prove the existence of an invariant measure, our aim is to
  construct a type of \textit{Lyapunov function} $\CV(p,q)$ such that  $\CV(p,q) \to \infty$ 
  as $|(p,q)| \to \infty$ and  $\CL \CV(p,q) \to -\infty$ as $|(p,q)| \to \infty$.   Given such 
  a function, the existence of an invariant measure will follow from the following 
  proposition which is a variant of the classical  Kryloff-Bogoliouboff construction \cite[p.~52]{KrylovBogo,Khas}:

  \begin{proposition}\label{abstractExistInvMeasure}
    Consider an SDE on $\R^n$ with smooth coefficients and denote its
    generator by $\CL$.  Assume that the SDE has global solutions and
    generates a Feller semigroup.  If there exists a smooth function
    $\CV\colon \R^n \to [0,\infty)$ such that the level sets
    $\{x\,:\,\CL \CV(x) \ge C\}$ are compact for every $C$, then the SDE
    possesses an invariant probability measure $\mu$.  Furthermore, the
    function $\CL \CV$ is integrable with respect to $\mu$ and $\int \CL
    \CV(x)\,\mu(dx) = 0$.
  \end{proposition}

  We will construct the function $\CV$ in steps by
  analyzing the dynamic in the limit of various energies being large
  and then draw inspiration from the structure of these limiting regimes to
  construct $\CV$. Operationally, we will make an initial guess for
  $\CV$ and then augment it by a series of correction terms.

  \begin{proof}[of Proposition \ref{abstractExistInvMeasure}]
    Denote by $x_t$ the solution to the SDE starting at some
    (deterministic) initial condition $x_0$.  Applying It\^o's formula
    to $\CV(x_t)$, we get
  \begin{equ}
  d\CV(x_t) =  \bigl(\CL \CV\bigr)(x_t)\,ds + dM(t)\;,
  \end{equ}
  for some continuous local martingale $M$. Therefore, there exists an increasing sequence $\tau_N$ of stopping times converging
  to $+\infty$ such that $M({t \wedge \tau_N})$ are martingales. 
  \begin{equ}
  \E \CV(x_{t\wedge \tau_N}) - \CV(x_0) - \E \int_0^{t\wedge \tau_N} \bigl(\CL \CV\bigr)(x_s)\,ds = 0\;. 
  \end{equ}
  Since $\CV$ is positive this shows that, for every $K>0$,
  \begin{equ}
   \E \int_0^{t\wedge \tau_N} \bigl(K - \bigl(\CL \CV\bigr)(x_s)\bigr)\,ds \le  \CV(x_0) + K t\;. 
  \end{equ}
  Taking $K$ large enough so that $K - \CL \CV \ge 0$, we can apply the monotone
  convergence theorem to take the limit $N \to \infty$ and obtain
  \begin{equ}
  - \E \int_0^t \bigl(\CL \CV\bigr)(x_s)\,ds \le \CV(x_0)\;.
  \end{equ}
  Now, by assumption the sets $A_R=\{ x: -\CL \CV(x)  \leq R\}$ are compact
  for all $R$. In particular, this implies that there exist a $K >0$ so
  that $-\CL \CV(x) + K  \geq 0$ for all $x$. Now observe that for $R>-K$ 
  \begin{align*}
    \frac1t \int_0^t \P\big( x_s \not\in A_R\big) ds &=   \frac1t \int_0^t \P\big(
    -\CL \CV(x_s) +K  > R+K\big) ds\\
    &\leq   \frac1t \int_0^t \frac{K - \E \CL \CV(x_s)}{R+K} ds \leq
    \frac{K + \CV(x_0)}{K+R}\,.
  \end{align*}
  Therefore, the sequence of measures $\mu_t$ defined for measurable sets $A$
  by
  \begin{equation*}
    \mu_t(A)=\frac1t \int_0^t \P_{x_0}(x_s \in A)\,ds
  \end{equation*}
  is tight. Hence the Kryloff-Bogoliouboff construction \cite[p.~52]{KrylovBogo,Khas}
  guarantees the existence of an invariant measure. The last statement
  then follows from Lebesgue's dominated convergence theorem.
  \end{proof}

  \begin{remark}
  We will actually be able to construct a positive smooth function $\CV$ with compact
  level sets that has the property that
  \begin{equ}
  \CL \CV \le C_1 - C_2 \CV^\alpha\;,
  \end{equ}
  for some positive constants $C_i$ and some (typically quite small)
  $\alpha \in (0,1]$. In this case, it is known \cite{DFG06SG} that one
  does not only have the existence of an invariant measure, but the
  transition probabilities converge towards it at rate
  $\CO(t^{-\alpha/(1-\alpha)})$ in the total variation distance.  We
  believe that this convergence actually takes place at a much
  faster rate, but such a statement is beyond our reach at the moment.
  See also \cite{Veret} for related results on subexponential mixing for SDEs.
  \end{remark}

  \subsection{Construction of the Lyapunov function}

  Recall the change of variables $(p,q) \mapsto (\bar p, \bar q)$ from
  Theorem~\ref{theo:decouple} that leads to an effective decoupled
  dynamic for the outside oscillators. In order to construct the
  Lyapunov function $\CV$, we proceed in two steps:
  \begin{claim}
  \item[1.] We gain good control over the dissipation of the energy
    stored in the outside oscillators. This will be the content of
    Proposition~\ref{prop:boundaries}.
  \item[2.] We use this in order to get control over the dissipation of the total energy of the system
  in Theorem~\ref{theo:wholething}.
  \end{claim}

  \begin{proposition}\label{prop:boundaries}
    There is a function $\CU_0$ equivalent
      to $\Hf(\bar p_0, \bar q_0)
    + \Hf(\bar p_2, \bar q_2)$ and such that, for every $m > 0$, there
    exist constants $C_m>0$ and $c_m > 0$ such that $\CL \CU_0^m \le C_m -
    c_m \CU_0^m$.
  \end{proposition}

  \begin{proof}
    Inspired by \eref{e:motionbar}, we define an \textit{effective Hamiltonian} $H_0$ by
    \begin{equ}
      H_0(p,q) = {p^2 \over 2} + {|q|^{2k}\over 2k} + {q^2 \over 2} + 1\;.
    \end{equ}
    (Note that $H_0$ is equivalent  to $\Hf$.) We set 
  \begin{equ}
    \CU_0 = H_0(\bar
    p_0, \bar q_0) + H_0(\bar p_2, \bar q_2) + \gamma (\bar p_0 \bar q_0
    + \bar p_2 \bar q_2) 
    \end{equ}
  for some constant $\gamma$ to be determined
    later. If $\gamma$ is sufficiently small and since $k \ge 1$, this
    function is indeed equivalent to $\bar E_0 + \bar E_2 = \Hf(\bar p_0, \bar q_0) + \Hf(\bar
    p_2, \bar q_2)$. Applying It\^o's formula to it, we get from
    \eref{e:motionbar} that
    \begin{equs}
      d \CU_0 &= \sum_{i = 0,2}\Bigl((\gamma - \gamma_i) \bar p_i^2 -
      \gamma \bigl(|\bar q_i|^{2k}+|\bar q_i|^{2}\bigr) -\gamma \gamma_i
      \bar p_i \bar q_i
      + \bigl(\bar p_i + \gamma \bar q_i \bigr) R_p^i \\
      &\qquad + \bigl(|\bar q_i|^{2k-2} \bar q_i +\gamma \bar p_i \bigr)
      R_q^i + \bigl(k-\hf\bigr) |\bar q_i|^{2k-2}
      \bigl|\Sigma_q^i\bigr|^2
      + \hf\bigl|\Sigma_p^i\bigr|^2 + \gamma \Sigma_p^i \Sigma_q^i\Bigr)\, dt \\
      &\qquad + \sum_{i = 0,2}\Bigl(\bigl(\bar p_i +\gamma \bar
      q_i\bigr)\Sigma_p^i + \bigl(|\bar q_i|^{2k-2} \bar q_i + \bar q_i
      +\gamma \bar p_i\bigr) \Sigma_q^i\Bigr)\, dw_i\;.
    \end{equs}
  Fixing $\gamma = \hf\min\{1, \gamma_0, \gamma_2\}$, we obtain for some
  constant $C$
    \begin{equs}
      d \CU_0 &\le - 2\gamma \CU_0 \,dt + C \sum_{i = 0,2}\Bigl( |\bar
      p_i \bar q_i|
      + \bigl|\bar p_i + \gamma \bar q_i \bigr| |R_p^i| \label{e:bounddU0} \\
      &\qquad + \bigl(|\bar q_i|^{2k-1} + |\bar p_i| \bigr) |R_q^i| + (1
      + |\bar q_i|^{2k-2} ) \bigl|\Sigma_q^i\bigr|^2 +
      \bigl|\Sigma_p^i\bigr|^2 \Bigr)\, dt + dM(t)\;.
    \end{equs}
  Here, $M$ is a continuous Martingale with quadratic variation bounded by
  \begin{equ}
    {d \scal{M}(t) \over dt} \le C\sum_{i = 0,2} \bigl((|\bar p_i|^2 +
    |\bar q_i|^2) | \Sigma_p^i|^2 + (1 + |\bar q_i|^{4k-2} + |\bar
    p_i|^2)|\Sigma_q^i|^2\bigr)\;.
  \end{equ}
  Using the notation and the bounds of Theorem~\ref{theo:decouple}, we have
  \begin{equs}[2]
    |\bar p_i \bar q_i| &\les \bar E_i^{{1 \over 2} + {1 \over 2k}}
    \;,&\quad
    \bigl|\bar p_i + \gamma \bar q_i \bigr| |R_p^i| &\les \bar E_i^{1-\delta} \;,\\
    \bigl(|\bar q_i|^{2k-1} + |\bar p_i| \bigr) |R_q^i| &\les \bar
    E_i^{1-\delta} + \bar E_{2-i}^{1-\delta}\;,&\quad (1 + |\bar
    q_i|^{2k-2} ) \bigl|\Sigma_q^i\bigr|^2 &\les \bar E_i^{{1 \over 2} -
      {1 \over k}}\;.
  \end{equs}
  Since all the powers appearing on the right hand sides of these bounds
  are strictly less than $1$, we have shown that
  \begin{equ}
  \CL \CU_0 \le C - \gamma \CU_0\;,
  \end{equ}
  for some (different) constant $C$.

  In order to bound the quadratic variation of $M$, note that one has
  \begin{equ}
    (|\bar p_i|^2 + |\bar q_i|^2) |\Sigma_p^i|^2 \les \bar E_i \;,\qquad
    (1 + |\bar q_i|^{4k-2} + |\bar p_i|^2)|\Sigma_q^i|^2 \les \bar
    E_i^{1 - {1 \over k}}\;.
  \end{equ}
  In particular, one has for some positive constant $C$
    \begin{equ}
      \scal{M}(t) \le C\int_0^t \bigl(1 + \CU_0(s)\bigr)\, ds\;.
    \end{equ}
    The claim now follows by applying It\^o's formula to $\CU_0^m$.
  \end{proof}

  \begin{remark}
    Observe that one could also apply It\^o's formula to $\exp(\theta
    \CU_0)$ for sufficiently small $\theta$ and obtain $\CL \exp (\theta
    \CU_0) \le C - \alpha \exp (\theta \CU_0)$.
  \end{remark}

  \begin{remark}
    If $k \ge 2$, then $\CU_0$ is also equivalent to $\Hf(p_0,
    q_0)+\Hf(p_2, q_2)$. This is \textit{not} the case when $k < 2$.
  \end{remark}

  We are now ready to prove the main theorem of this section. Recall that $H$
  is the total Hamiltonian defined in \eref{e:defH}.

  \begin{theorem}\label{theo:wholething}
    Consider the equations of motion \eref{e:threeosc} with $k > 3/2$.
    Then there exists a function $\CV$ and constants $c, C>0$ such that
    $\CV \ge cH^\alpha - C$ and such that $\CL \CV \le C - cH^{\alpha'}$
    for some exponents $\alpha$ and $\alpha'$.
  \end{theorem}

  \begin{proof}
    The idea is to work ``modulo powers of $\CU_0$.'' Assume that we can
    find a function $\CU_1$ such that
    \begin{equ}[e:U1]
      \CU_1 \ge cH^\alpha - C \CU_0^{N}\;,\qquad \CL \CU_1 \le -
      cH^{\alpha'} + C \CU_0^{N'}\;,
    \end{equ}
    for some positive exponents $\alpha$, $\alpha'$ and some (possibly
    very large) exponents $N$, $N'$. We claim that it then suffices to
    take $\CV = \CU_1 + \CU_0^{N+N' + 1}$.  Note that
    \begin{equ}
  \CV \ge c H^\alpha - C\CU_0^N + \CU_0^{N+N'+1} \ge c H^\alpha - C'\;,
  \end{equ}
  for some constant $C'$,
  and so $\CV$ grows at infinity (has compact level sets).
    It then follows from
    Proposition~\ref{prop:boundaries} and from \eref{e:U1} that
  \begin{equ}
    \CL \CV \le - c H^{\alpha'} + C \CU_0^{N'} + C_{N+N'+1} - c_{N+N'+1}
    \CU_0^{N+N'+1} \le C'' - cH^{\alpha'}\;,
  \end{equ}
  for some constant $C''$, which is the desired bound. It therefore
  remains to construct a function $\CU_1$ satisfying \eref{e:U1}.

  The starting point for the construction of $\CU_1$ is to apply It\^o's
  formula to $H^n$. Since $dH = \sum_{i=0,2} \bigl(-\gamma_i p_i\,dt +
  \hf \sigma_i^2\, dt + p_i\sigma_i\, dw_i\bigr)$, we have
    \begin{equ}
      dH^n = nH^{n-2} \sum_{i=0,2} \Bigl(H\bigl({\sigma_i^2 \over 2} -
      \gamma_i p_i^2\bigr)\, dt + \hf (n-1)\sigma_i^2 p_i^2\, dt + H
      \sigma_i p_i dw_i \Bigr)\;.
      \label{e:ItoHn}
    \end{equ}
    At this stage, we have to distinguish two cases, as in the proof of
    Theorem~\ref{theo:decouple}:
    \begin{claim}
    \item If $3/2 < k < 2$, then the motion of the interior oscillator
      induces large (\ie going to infinity with the energy of the middle
      oscillator) fluctuations in the values of $p_0$ and $p_2$.
      Therefore, there will always be energy dissipation.
    \item If $k \ge 2$, then the motion of the interior oscillator
      induces bounded (or small) fluctuations in the values of $p_0$ and
      $p_2$.  In this case, one can subtract a compensator so that only
      fluctuations remain.  There will be a few error terms that seem to
      be larger than the dominant one, but they can hopefully just be
      eliminated order by order.
    \end{claim}
  \noindent\textit{The case  $3/2 < k < 2$:} In contrast to the
  arguments in the proof of Theorem~\ref{theo:decouple}, this is the ``easy'' case for this part of the proof.  It
  follows from the proof of Theorem~\ref{theo:decouple} that in this
  case $\bar p_i^2$ is equivalent to
  $\hat p_i^2$ with $\hat p_i = p_i - \Phi(p_1,q_1)$. 
  Since $2\gamma_i \hat p_i \Phi \le \hf \gamma_i \Phi^2 + 2 \gamma_i \hat p_i^2$,
  this allows us to obtain for $\CL H^n$ the bound
  \begin{equ}
    \CL H^n \le n H^{n-1} \bigl(C_n + \gamma_0 \hat p_0^2 + \gamma_2
    \hat p_2^2 - \hf (\gamma_0 + \gamma_2) \Phi^2\bigr)\;,
  \end{equ}
    for some positive constant $C_n$. This in turn implies that
  \begin{equ}[e:boundLHn]
    \CL H^n \le H^{n-1} \bigl(C_1 + C_2 \CU_0 - \textstyle{n \over
      2}(\gamma_0 + \gamma_2) \Phi^2\bigr)\;,
  \end{equ}
  for some constants $C_i$. Observe that one can use a similar
  calculation to obtain the bound
  \begin{equ}[e:boundabsLHn]
    |\CL H^n| \le H^{n-1} \bigl(C_1 + C_2 \CU_0 + C_3 \Phi^2\bigr)\;,
  \end{equ}
  for some possibly different constants $C_i$.

  Note that $\Phi^2$ scales like $\Hf^{{2\over k} - 1}$ which goes to
  infinity at high energies.  We can thus find a positive constant
  $\kappa_k$ (which was introduced in \eref{e:defkappak}) and a function $\CR'$ such that $\Phi^2 - \kappa_k \Hf^{{2\over
      k} - 1} + \CR'(\Hf)$ averages out to $0$ in the sense of
  Proposition~\ref{prop:scale}. We define $\Psi$ to be the solution to
  the Poisson equation
    \begin{equ}[e:defPsi]
      X_{\Hf} \Psi = \Phi^2 - \kappa_k \Hf^{{2\over k} - 1} + \CR'(\Hf)\;.
  \end{equ}
  This function is smooth since $\CR'$ is precisely such that $\Psi = 0$
  in a neighborhood of the origin. The function $\Psi$ then scales like
  $\Hf^{{5 \over 2k} - {3\over 2}}$. In particular, $\Psi$ and $Q \d_P
  \Psi$ are bounded.

    With all these preliminaries done, we define
    \begin{equ}
      \CU_1 = H^n + \textstyle{n\over 2}(\gamma_0 + \gamma_2) H^{n-1}
      \Psi\;.
    \end{equ}
    Applying It\^o's formula to $\CU_1$ and using \eref{e:boundLHn},
    \eref{e:boundabsLHn}, and \eref{e:defPsi}, we obtain
  \begin{equs}
    \CL \CU_1 &\le C H^{n-1} \bigl(1 + \CU_0\bigr) + C |\Psi| H^{n-2}
    \bigl(1 + \CU_0 + \Phi^2\bigr) \\
    &\quad + \textstyle{n\over 2}(\gamma_0 + \gamma_2) H^{n-1}
    \bigl(-\kappa_k E_1^{{2\over k} - 1} + (q_0 + q_2 - 2q_1)\d_P
    \Psi\bigr)\;.
  \end{equs}
  Since $|\Psi|\Phi^2 \les H$ and $|(q_0 + q_2 - 2q_1)\d_P \Psi| \les
  \CU_0$, we have
  \begin{equ}[e:boundLU1]
    \CL \CU_1 \le -\textstyle{\kappa_k n\over 2}(\gamma_0 + \gamma_2)
    H^{n-1} \kappa_k E_1^{{2\over k} - 1} + C H^{n-1} \CU_0\;,
  \end{equ}
  for some constant $C$.

  We can check that $\CU_1$ satisfies the first bound in \eref{e:U1},
  since $\Psi$ grows slower than $H$ and $\CU_0$ is bounded from below
  by a positive constant.  However, it is not clear \textit{a priori}
  from this expression that it $\CU_1$ satisfies the second bound in
  \eref{e:U1}.

  One can check that
  \begin{equ}
    \CU_0 \ge c(E_0 + E_2) - C\Phi^2\;,
  \end{equ}
  for some positive constants $c$ and $C$. Furthermore, $\Phi^2$ scales
  like $\Hf^{{2\over k}-1}$ which is strictly smaller than $\Hf$, so
  that
  \begin{equ}
  \CU_0 + E_1 \ge cH \;.
  \end{equ}
  In particular, there exists a constant $C$ such that
  \begin{equ}
  H^{{2\over k}-1} \le C \bigl(E_1^{{2\over k}-1} + \CU_0^{{2\over k}-1}\bigr)\;. 
  \end{equ}
  Inserting this into \eref{e:boundLU1}, we obtain
  \begin{equ}
  \CL\CU_1 \le - cH^{n + {2\over k}-2} + C H^{n-1} \CU_0\;,
  \end{equ}
  for some constants $c$ and $C$. Since ${2\over k} -1 > 0$, this shows
  the existence of a (sufficiently large) power $N$ such that
  \begin{equ}
  \CL\CU_1 \le - cH^{n + {2\over k}-2} + C \CU_0^N\;,
  \end{equ}
  which is the desired bound.

  \medskip \noindent\textit{The case $k\geq 2$:} The reasoning above
  only worked for $k < 2$ since in this case one has $n + {2\over k} - 2
  > n-1$. When $k\geq 2$ the situation is slightly more delicate.  Here,
  as before, the energy dissipation mechanism comes from fluctuations of
  $p_0^2$ and $p_2^2$ around their mean. However, the amplitude of these
  oscillations now \textit{decreases} as the energy stored in the middle
  oscillator increases, so that we have to be more careful. In
  particular, $p_i^2$ can no longer be treated as a perturbation with
  respect to $\Phi^2$, as we did in the previous case.

  The expression in \eref{e:ItoHn} suggests that, in order to extract
  these small fluctuations, we should compensate $H^n$ by subtracting
    \begin{equ}
      \Xi = nH^{n-1}\big[ H_0(\bar p_0, \bar q_0)+ H_0(\bar p_2, \bar
      q_2) \big]\;,
    \end{equ}
    where the ``bar'' variables are as in Theorem~\ref{theo:effective2}.
    Note that here, and anywhere from this point on, we use the
    variables $(\bar p_i, \bar q_i)$ from Theorem~\ref{theo:effective2}
    and \textit{not} from Theorem~\ref{theo:decouple}. This is because
    we require a sufficiently good effective dynamics so that the error
    terms are small with respect to $\Phi^2$. The bounds obtained in
    Theorem~\ref{theo:decouple} are not sufficiently small for that.
    Setting as before $\bar E_i$ as a shorthand for $\Hf(\bar p_i, \bar
    q_i)$, one has
  \begin{equs}
    \CL \Xi &= n H^{n-2} \sum_{i=0,2} \Bigl(H \bigl({\sigma_i^2 \over 2}
    - \gamma_i \bar p_i^2\bigr)
    + (n-1)\bar E_i \bigl({\sigma_i^2 \over 2} - \gamma_i p_i^2\bigr) \\
    & \quad + \hf (n-1)(n-2) \bar E_i H^{-1} \sigma_i^2 p_i^2 + (n-1)
    \sigma_i^2 p_i \bar p_i \\ 
    &\quad + H \bigl(\bar p_i R_p^i + \bar q_i \bigl(|\bar q_i|^{2k-2} +
    1\bigr) R_q^i \bigr)\Bigr)
  \end{equs}
  The important term in this expression is $n H^{n-1}
  \bigl(\textstyle{\frac{\sigma_0^2 +\sigma_2^2}2} - \gamma_0 \bar p_0^2
  - \gamma_2 \bar p_2^2 \bigr)$, all the other terms will be treated as
  perturbations. Setting $\CU_1^0 = H^n - \Xi$, we obtain
    \begin{equ}
      \CL \CU_1^0 = - nH^{n-1} \sum_{i=0,2} \gamma_i (p_i^2 - \bar  p_i^2) + \CE\;.
    \end{equ}
  with
  \begin{equs}
    \CE &= n H^{n-2} \sum_{i=0,2} \Bigl( - (n-1)\bar E_i
    \bigl({\sigma_i^2 \over 2} - \gamma_i p_i^2\bigr) - H \bigl(\bar p_i
    R_p^i + \bar q_i \bigl(|\bar q_i|^{2k-2} + 1\bigr)
    R_q^i \bigr) \\
    & \quad - \hf (n-1)\sigma_i^2 \bigl((n-2) \bar E_i H^{-1} p_i^2 - 2
    p_i \bar p_i + p_i^2\bigr)\Bigr)\;.
  \end{equs}
  We can check that one has $|\CE| \les H^{n -2 + {3 \over 2k}} \CU_0^2$ so that,
  for every $\delta>0$, there exist constants $C$ and $N$ such that
  the error term $\CE$ is bounded by $H^{n + {3 \over 2k} -2 + \delta} + C \CU_0^N$.

    At this stage, we use the explicit expression for $\bar p_i$ to rewrite this bound as
    \begin{equ}
      \CL \CU_1^0 \le - nH^{n-1} \sum_{i=0,2} \gamma_i \bigl(\Phi^2 +
      \gamma_i^2 \bigl(\Phi^{(2)}\bigr)^2- 2\gamma_i \Phi \Phi^{(2)} +
      2\bar p_i(\gamma_i \Phi^{(2)} - \Phi)\bigr) + \CE\;.
    \end{equ}
    This puts us now in a situation similar to \eref{e:boundLHn}, with
    the difference that, up to powers of $\CU_0$, the error term $\CE$
    is of order $H^{n-2 + {3 \over 2k} + \delta}$ instead of being of
    order $H^{n-1}$. Since $\Phi^2$ is larger than $\Hf^{{3 \over 2k}
      -1}$, this is the feature that will allow us to obtain the
    required bound.

    Since $\Phi^{(2)}$ scales like $\Hf^{{3 \over 2k}-1}$, we obtain in
    the same way as in the proof of Theorem~\ref{theo:effective2} that
    $|\Phi^{(2)}| \les H^{{3 \over 2k} - 1} \CU_0^{1-{3 \over 2k}}$.
    This shows that
    \begin{equ}[e:LU01]
      \CL \CU_1^0 \le - nH^{n-1} \sum_{i=0,2} \gamma_i \bigl(\Phi^2 -
      2\bar p_i \Phi \bigr) + \CE_2\;,
    \end{equ}
    where the error term $\CE_2$ satisfies the bound $|\CE_2| \les H^{n
      + {3 \over 2k} -2 + \delta} + C \CU_0^N$ as before\footnote{The
      interested reader can check that expanding the square around $p_i$
      instead of $\bar p_i$ would lead to troublesome terms.}.

    Since these terms oscillate very rapidly, we would like to replace
    them by their averaged effect over one period of the middle
    oscillator. To leading order, the terms $p_i \Phi$ and $\Phi^2$ will
    average out to $0$ and $\kappa_k E_1^{{2 \over k}-1}$ respectively.
    The latter contribution will turn out to be the dominant term
    leading to an overall dissipation of energy.  Defining $\Psi$ and
    $\Phi^{(2)}$ as in \eref{e:defPsi} and \eref{e:defPhi2}
    respectively, we finally set
  \begin{equ}[e:defU1]
    \CU_1 = \CU_1^0 + n H^{n-1} \sum_{i=0,2} \gamma_i \bigl(\Psi - 2\bar
    p_i \Phi^{(2)}\bigr)\;.
  \end{equ}
  Our investigation of $\CL\CU_1$ begins with
  \begin{equs}
    \CL H^{n-1} \Psi &= -H^{n-1} \bigl(\kappa_k E_1^{{2 \over k}-1} -
    \Phi^2\bigr) + H^{n-1}\Bigl (\CR'(E_1) + (q_0 + q_2 - 2q_1)\d_P
    \Psi\Bigr) \\
    &\qquad + (n-1)H^{n-3} \Psi\sum_{i=0,2} \Bigl(H\bigl({\sigma_i^2
      \over 2} - \gamma_i p_i^2\bigr)
    + \hf (n-2)\sigma_i^2 p_i^2 \Bigr) \\
    & = -H^{n-1} \bigl(\kappa_k E_1^{{2 \over k}-1} - \Phi^2\bigr) + I_9 +
    I_{10}\;.
  \end{equs}
  Since the function $\CR'$ has compact support, we have for example
  $|\CR'(E_1)| \les E_1^{-1}$.  Using the trick from
  Lemma~\ref{lem:trivial}, this yields $|\CR'(E_1)| \les H^{-1} \CU_0$.
  This allows us to obtain the bound
  \begin{equs}
    |I_9| &\les H^{n-2} \CU_0 + H^{n-1 + {1 \over 2k}} E_1^{{5 \over 2k}
      - 2} \les H^{n-2} \CU_0 + H^{n-3 + {3 \over k}} \CU_0^{2-{5 \over
        2k}}\;.
  \end{equs}
  Since $\Psi$ is bounded, one furthermore has
  \begin{equ}
    |I_{10}| \les H^{n-2} \CU_0 + H^{n-3} \CU_0 \;.
  \end{equ}
  Turning to the second term in $\CL\CU_1$, we have
  \begin{equs}
    \CL H^{n-1} \bar p_i \Phi^{(2)} &= H^{n-1}\bar p_i \Phi +
    H^{n-1}\bar p_i (q_0 + q_2 - 2q_1)\d_P \Phi^{(2)} \\
    &\quad + (n-1)H^{n-3} \bar p_i \Phi^{(2)}\sum_{j=0,2}
    \Bigl(H\bigl({\sigma_j^2 \over 2} - \gamma_j p_j^2\bigr)
    + \hf (n-2)\sigma_j^2 p_j^2 \Bigr) \\
    &\quad+ H^{n-1}\Phi^{(2)} \bigl(- \bar q_i|\bar q_i|^{2k-2} - \bar
    q_i + R_p^i - \gamma_i \bar p_i\bigr) + (n-1)H^{n-2}\sigma_i^2 p_i
    \Phi^{(2)} \\
    & = H^{n-1}\bar p_i \Phi + I_{11} + I_{12}+ I_{13}+ I_{14}\;.
  \end{equs}
  Using the bound on $R_p^i$ from \eref{e:errorbounds2}, the error terms
  appearing in this expression can be bounded by
  \begin{equs}
    |I_{11}| &\les H^{n-1+{1 \over 2k}} \CU_0^{{1 \over 2}} E_1^{{3
        \over 2k} - {3 \over 2}} \les H^{n-{5 \over 2}+{2 \over k}}
    \CU_0^{2-{3 \over 2k}}\;,\\
    |I_{12}| &\les H^{n-2} \CU_0^{{3 \over 2}} + H^{n-3} \CU_0^{{3 \over 2}} \;,\\
    |I_{13}| &\les H^{n-1} E_1^{{3 \over 2k} - 1} \bigl(\CU_0 + \CU_0^2
    H^{{3 \over 2k} - 1}\bigr) \les H^{n -2 + {3 \over 2k}} \CU_0^{3 -
      {3 \over 2k}}\;,\\
    |I_{14}| &\les H^{n-2} \CU_0^{{1 \over 2}}\;.
  \end{equs}
  By Young's inequality, it can be checked that, for every $\delta > 0$
  there exist constants $C$ and $N$ such that $|I_j| \le H^{n -2 + {3
      \over 2k} + \delta} + C \CU_0^N$ for $j=9,\ldots,14$. Inserting
  these bounds into \eref{e:LU01} and \eref{e:defU1}, this shows that
  \begin{equ}
    \CL \CU_1 \le -n \kappa_k (\gamma_0 + \gamma_2) H^{n-1} E_1^{{2 \over
        k}-1} + H^{n -2 + {3 \over 2k} + \delta} + C \CU_0^N\;.
  \end{equ}
  Since ${2 \over k} - 1 \le 0$ and since $n-2+{2 \over k} > n -2 + {3
    \over 2k} + \delta$ for sufficiently small $\delta$, we finally
  obtain
  \begin{equ}
    \CL \CU_1 \le -\hf n \kappa_k (\gamma_0 + \gamma_2) H^{n-2 + {2 \over
        k}} + C \CU_0^N\;,
  \end{equ}
  for a possibly different constant $C$. This is precisely the bound
  \eref{e:U1} which was the missing piece to complete the proof of
  Theorem~\ref{theo:wholething}, the principal result of this article.
  \end{proof}

\def\Rom#1{\uppercase\expandafter{\romannumeral #1}}\def\u#1{{\accent"15
  #1}}\def\cprime{$'$}

\end{document}